\documentclass[journal]{IEEEtran}
\usepackage[T1]{fontenc}% optional T1 font encoding
\usepackage[matit]{echodefs}
\usepackage[textsize=footnotesize]{todonotes}

\usepackage{algorithm,algpseudocode}
\usepackage{graphicx}
\usepackage{amsmath}
\usepackage{subcaption}
\usepackage{color, colortbl}
\definecolor{Gray}{gray}{0.9}
\usepackage{booktabs}%\renewcommand{\arraystretch}{1.3}

\usepackage{cleveref}

\newcommand{\EDMgram}{\ensuremath{\mathcal{K}}}

\renewcommand{\vzero}{\ensuremath{\vec{\mathit{0}}}}

\newcommand{\trace}{\ensuremath{\mathop{\mathrm{Tr}}}}

\newcommand{\mLambda}{\bm{\mathit{\Lambda}}}
\renewcommand{\diag}{\mathop{\mathrm{diag}}}
\newcommand{\affdim}{\ensuremath{\mathrm{affdim}}}

\newcommand{\mTheta}{\mat{\Theta}}
\newcommand{\embdim}{{d}}

\begin{document}

\title{Kinetic Euclidean Distance Matrices}
\author{Puoya Tabaghi,
        Ivan~Dokmani\'c,
        Martin Vetterli
        \thanks{This work was presented at the ``Distance Geometry:
Inverse problems between geometry and optimization'' workshop in November 2017 (\url{https://www.lix.polytechnique.fr/\~liberti/dg17/}).}
} % <-this % stops a space

% make the title area
\maketitle

% As a general rule, do not put math, special symbols or citations
% in the abstract or keywords.
\begin{abstract}
Euclidean distance matrices (EDMs) are a major tool for localization from distances, with applications ranging from protein structure determination to global positioning and manifold learning. They are, however, static objects which serve to localize points from a snapshot of distances. If the objects move, one expects to do better by modeling the motion. In this paper, we introduce Kinetic Euclidean Distance Matrices (KEDMs)---a new kind of time-dependent distance matrices that incorporate motion. The entries of KEDMs become functions of time, the squared time-varying distances. We study two smooth trajectory models---polynomial and bandlimited trajectories---and show that these trajectories can be reconstructed from incomplete, noisy distance observations, scattered over multiple time instants. Our main contribution is a semidefinite relaxation (SDR), inspired by SDRs for static EDMs. Similarly to the static case, the SDR is followed by a spectral factorization step; however, because spectral factorization of polynomial matrices is more challenging than for constant matrices, we propose a new factorization method that uses anchor measurements. Extensive numerical experiments show that KEDMs and the new semidefinite relaxation accurately reconstruct trajectories from noisy, incomplete distance data and that, in fact, motion improves rather than degrades localization if properly modeled. This makes KEDMs a promising tool for problems in geometry of dynamic points sets.
\end{abstract}

% Note that keywords are not normally used for peerreview papers.

\IEEEpeerreviewmaketitle

\section{Introduction}

% \tableofcontents

\IEEEPARstart{T}{h}e famous distance geometry problem (DGP) \cite{Liberti:2012ut} asks to reconstruct the geometry of a point set from a subset of interpoint distances. It models a wide gamut of practical problems, from sensor network localization \cite{Biswas:2006cm,Krislock:2010ga,Bai:2015tg} and microphone positioning \cite{Dokmanic:2015cz, Dokmanic:2012ii, Dokmanic:2013dz, Parhizkar:2014kn} to clock synchronization \cite{Singer:2008kh, Liberti:2016tr}, to molecular geometry reconstruction from NMR data \cite{Hendrickson:1995bg, Hendrickson:1990tg}. Among the most successful vehicles for the design of DGP algorithms are the Euclidean distance matrices (EDM) \cite{Dokmanic:2015eg}.

EDMs model static objects. When things move, they characterize a snapshot of the interpoint distances and the point set geometry. It seems intuitive that with a good model for the trajectories one should be able to leverage the motion and improve trajectory estimation.

In this work we make a first step towards distance matrices for moving points, which we call Kinetic EDMs (KEDMs) inspired by the notion of kinetic data structures \cite{Guibas:1998ww} for moving points. KEDMs are a generalization of EDMs whose entries now become functions of time. We show how by using KEDMs we can neatly address the kinetic distance geometry problem (KDGP), a natural generalization of the classical, static distance geometry problem (DGP) defined in \Cref{sec:kdgp}. Unlike with the static DGP, in order to make the kinetic version well posed, we must constrain the point trajectories to belong to a class of functions, for example polynomial trajectories of a bounded degree. Informally, we ask the following question: suppose a set of points move according to a known trajectory model. At given time instants we measure a subset of pairwise distances; the subset can change between measurements and it may be too small to allow localization at any time alone. Can we systematically localize the points and reconstruct trajectories by exploiting the trajectory model?

Localization of dynamic point sets from  distance measurements finds applications whenever objects move. Robot swarms, for example, often must localize autonomously \cite{cornejo2015distributed}, especially in remote situations such as extraterrestrial exploration \cite{matthaei2013swarm} or deep-water missions \cite{jaffe2017swarm}. Related applications exist in environmental monitoring, for example for dynamic sensor networks composed of river-borne sensing nodes \cite{Wu:2015br}. An important application of localization of moving objects is in global positioning with satellites where both the satellites and the users move. Applications are emerging where sensing is opportunistic and the positions of reference objects are not known \cite{Morales:2016xx}; in \Cref{sec:gps}, we present a simulated example of global positioning with unknown satellite trajectories. This problem is further related to simultaneous localization and mapping (SLAM) \cite{DurrantWhyte:2006ev,SLAMK2016}. Kinetic distance geometry problems are common in computer vision. Examples are action recognition from dynamic interjoint distance skeleton data \cite{hernandez20173d} and more generally data structures for describing kinetic point sets \cite{Guibas:1998ww}. Applications in multi-robot coordination, crowd simulations, and motion retargeting are explored in \cite{mucherino2017approach, mucherino2018application}, where the authors introduce the \emph{dynamical distance geometry problem} (dynDGP).\footnote{Though related, the dynDGP is rather different from our KDGP.} Even in applications to proteins and molecules, the atoms move (for example, proteins fold) in specific ways \cite{wiki:xxx}.

The study of distance geometry and EDMs goes back to the work of Menger \cite{Menger:1928bc}, Schoenberg \cite{Schoenberg:1935dk},
Blumenthal \cite{Blumenthal:1953ie}, and Young and Householder \cite{Browne1987}. Gower  derived numerous results on EDMs \cite{Gower:1982wm, Gower:1985un} including a complete rank characterization \cite{Gower:1985un}. An extensive treatise on EDMs with many original results and an elegant characterization of the EDM cone was written by Dattorro \cite{Dattorro:2011wa}; \cite{Dokmanic:2015eg} is a tutorial-style introduction to EDMs.

A large class of approaches to point set localization from distance measurements rely on semidefinite programming \cite{Krislock:2012xx,So:2007cz}. Namely, the localization problem is written in terms of the Gram matrix of the point set which leads to a rank-constrained semidefinite program. The rank constraint is often relaxed to arrive at a semidefinite relaxation which is a convex optimization problem and can be solved using standard tools. 

We take inspiration from these approaches and show how trajectory localization can also be formulated as a semidefinite program, thus answering the above question in the affirmative. Concretely, we show that the parameters of chosen trajectory models can be recovered by a semidefinite program and a tailor-made alignment procedure akin to Procrustes analysis. The latter can be interpreted as spectral factorization of semidefinite polynomial matrices with side information, and our developments rely on the related spectral factorization results \cite{tabaghi2019real}.

As we will show, a major difference from the static case is that the (time-varying) distance information does not specify the point set up to a rigid transformation. The larger class of ambiguities can nevertheless be resolved through the use of anchor points.
   
We show through extensive computational experiments that our semidefinite relaxation indeed works as expected and that with an appropriate trajectory model we can reduce the number of measurements \emph{per time instant} well below that minimally required for localization in the static case.

\subsection{Paper Outline} % (fold)
\label{sub:paper_outline}

In \Cref{sec:kdgp} we extend the definition of the distance geometry problem (DGP) to its kinetic version and review the essential facts about Euclidean distance matrices and associated semidefinite programs. Next, in Section \ref{sec:trajectories_and_basis_gramians}, we introduce the two trajectory models we will use throughout the paper---polynomial and bandlimited---and show how to write the corresponding polynomial Gram matrices in a form suitable for semidefinite programming. We use this form (given in terms of the so-called basis Gramians) to formulate a semidefinite relaxation in \Cref{sec:kinetic_gramian_estimation}.
The solution to the SDP, however, only gives us \textit{distance} trajectories. 
To convert them to point trajectories, we need known and new results on spectral factorization of polynomial matrices developed in Section \ref{sec:sf}. Finally, we provide simulation results and list several promising directions for future work.

\subsection{Reproducible Research}

We adhere to the philosophy of reproducible research: documented code and data to reproduce all experiments is available online.\footnote{\url{https://github.com/swing-research/kedm/}}

\section{Static and Kinetic Distance Geometry Problems} 
\label{sec:kdgp}

\begin{figure*}[t]
\centering
\includegraphics[width=.95\linewidth]{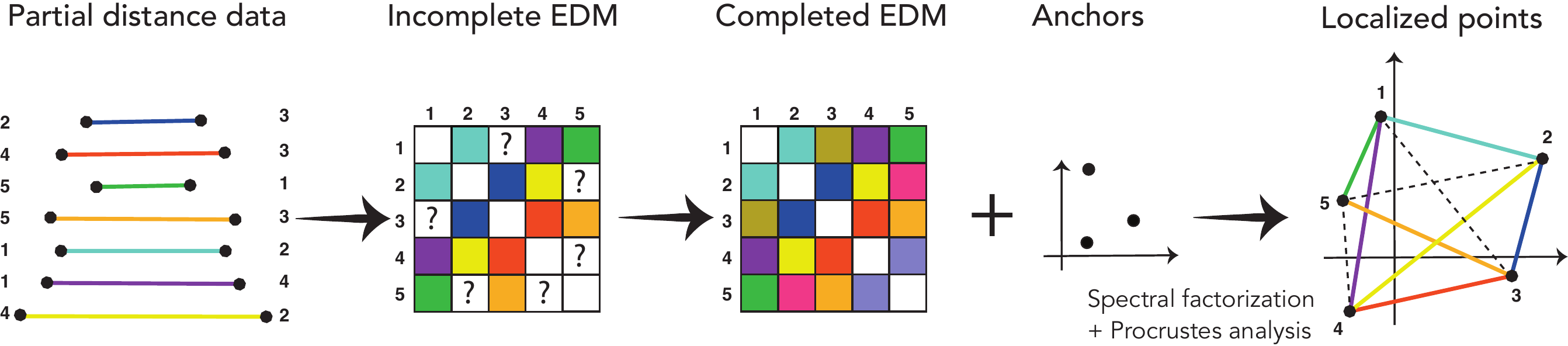}
\caption{The objective of DGP is to find an embedding for a given partial pair-wise distance data. This can be done in two steps: a) Completing EDM associated with the measurements, i.e. estimating the missing measurements and b) Estimate an embedding and using anchor points to resolve the rigid transformation ambiguity, discussed in \Cref{sec:solving_dgp_edms}.}
\label{fig:dgp-edm-illustration}
\end{figure*}

We begin by introducing the classical distance geometry problem (DGP) and then formulate its generalization to moving points. We also discuss an EDM-based approach to the DGP with noisy and incomplete distances. 

The DGP can be informally stated as follows: find the $\embdim$-dimensional locations $\set{\vx_n \in \R^{\embdim}}_{n=1}^N$ of a set of points, given a subset of possibly noisy pairwise distances $\{d_{mn}: 1 \leq m < n \leq N\}$. We will work only with Euclidean distances so that $d_{mn} = \norm{\vx_m - \vx_n}$.

An elegant formalization can be made in graph-theoretic terms. Consider a graph $G = (V, E)$ whose vertex set $V$ corresponds to the points $\set{\vx_n}_{n=1}^{N}$. The edge set $E$ tells us which distances are measured and which are not. Given two vertices $u, v \in V$ and the corresponding undirected edge $e = \set{u, v}$, we have $e \in E$ if and only if the distance between $u$ and $v$ is known. Let $f : E \to \R^+$ be the weight function that assigns those known, measured distances to edges. Then we can pose the following problem \cite{Liberti:2012ut}:
\begin{problem}[Distance Geometry Problem] \label{prob:dgp}
Given an integer $\embdim > 0$ (the ambient dimension) and an undirected graph $G = (V, E)$ whose edges are weighted by a non-negative function $f : E \rightarrow \R^{+}$ (distance), determine whether there is a function $x : V \rightarrow \R^{\embdim}$ such that
\begin{equation*}
\text{for all}~\{u, v\} \in E~\text{we have}~\|x(u) - x(v)\| = f(\{u, v\}).
\end{equation*}
\end{problem}
The function $x$ which assigns coordinates to vertices is called an embedding or a realization of the graph in $\R^\embdim$. Of course, in practice the measurements are corrupted by measurement errors, and the goal is to minimize some notion of discrepancy between the measured distances and the distances induced by our estimate; for example:
\begin{equation} 
\underset{x : V \to \R^\embdim}{\text{minimize}}  \sum_{\set{u, v} \in E} \left( \norm{x(u) - x(v)} - f(\set{u, v}) \right)^2.
\label{eq:dgp_in_x}
\end{equation}
Section \ref{sec:solving_dgp_edms} explains how to use EDMs to proceed in this case. Figure \ref{fig:dgp-edm-illustration} illustrates the DGP with an intermediate step of constructing an EDM. The EDM can be interpreted as a weighted adjacency matrix with weights being squared distances.

In this paper, we address distance geometry problems when the points move and the set of measured distance changes over time. Instead of localizing the points only at the measurement times, our goal will be to estimate  entire trajectories for all times. To make this problem well posed we must introduce a class of admissible continuous trajectories $\mathcal{X}$. Then, we can formulate the following kinetic version of Problem \ref{prob:dgp}:

\begin{problem}[Kinetic Distance Geometry Problem] \label{prob:kdgp}
Given an embedding dimension $\embdim > 0$, a set of $T$ sampling times $\mathcal{T} = \set{t_1, \ldots, t_T} \subset \R$, and a sequence of undirected graphs $G_i = (V, E_i)$ whose edges are weighted by non-negative functions $f_i : E_i \rightarrow \R^{+}$, for $i \in \set{1, \ldots, T}$, determine whether there is a function $x : V \times \R \rightarrow \R^{d} \in \mathcal{X}$ such that for all $t_i \in \mathcal{T}$ we have:
\begin{equation*}
\|x(u, t_i) - x(v, t_i)\| = f_i(\{u, v\})~\text{for all}~\{u, v\} \in E_i,
\end{equation*}
where $\mathcal{X}$ is the set of admissible trajectories.
\end{problem}

Figure \ref{fig:kdgp-kedm-illustration} illustrates the KDGP for four trajectories. One way to interpret KDGP is as a sequence of static DGPs with additional information about sampling times and trajectory model. Indeed, the KDGP can be seen a a nonlinear spatio--temporal sampling problem, with the nonlinear samples (distances) spread in space in time. A natural question is whether we can compensate for a reduction in the number of spatial samples by oversampling in time. We answer this question in affirmative in Section \ref{sec:results}.

The first step is to estimate the continuous KEDM from samples distributed in space and time; this is discussed in Section \ref{sec:kinetic_gramian_estimation}. The second step is to use information about the absolute positions of a set of anchor points in order to assign absolute locations to trajectories; this is discussed in Section \ref{sec:sf}. This step is more challenging than for the usual EDMs. 

\subsection{Solving the Distance Geometry Problem by EDMs} 
\label{sec:solving_dgp_edms}

It is useful to recall the EDM-based approach to the DGP. %on which we will build our KEDM algorithms. 
Ascribe the coordinates of $N$ points in a $D$-dimensional space to the columns of matrix $\mX \in \R^{\embdim \times N}$,  $\mX = [\vx_1, \ \vx_2, \ \cdots,\ \vx_N]$. The squared distance between $\vx_i$ and $\vx_j$ is 
\begin{equation*}
d_{ij}^2 = \norm{\vx_i - \vx_j}^2 = \norm{\vx_i}^2 - 2 \vx_i^\T \vx_j + \norm{\vx_j}^2,
\end{equation*}
from which we can read out the equation for the EDM $\mD = (d_{ij}^2)$ as
\begin{equation}
    \label{eq:edm_gram_assemble}
    \mD = \EDMgram(\mG) \bydef \diag(\mG) \vone^\T  - 2 \mG + \vone \diag(\mG)^\T,
\end{equation}
where $\vone$ denotes the column vector of all ones, $\mG$ is the Gram matrix $\mG = \mX^\T \mX$, and $\diag(\mG)$ is a column vector of the diagonal entries of $\mG$. We see that the EDM of a point set is a linear function of its Gram matrix. Reformulating the problem in terms of the Gram matrix is beneficial because it will lead to a semidefinite program. If we can find the Gram matrix,  the point set can be obtained by an eigenvalue decomposition.

To see how, let $\mG = \mU \mLambda \mU^\T$, where $\mLambda = \diag(\lambda_1, \ldots, \lambda_N)$ with all eigenvalues $\lambda_i$ non-negative, and $\mU$ orthonormal, as $\mG$ is a symmetric positive semidefinite matrix. Assume that the eigenvalues are sorted in decreasing order %the order of decreasing magnitude, 
$\lambda_1 \geq \lambda_2 \geq \cdots \geq \lambda_N$. Then we can estimate the point set as $\wh{\mX} \bydef \big[ \diag \big( \sqrt{\lambda_1},\ \ldots,\sqrt{\lambda_\embdim}\, \big), \ \vzero_{\embdim \times (N-\embdim)} \big] \mU^\T$. Since the EDM only specifies the points up to a rigid transform, $\wh{\mX}$ will be a rotated, reflected and translated version of $\mX$.	

One way to estimate $\mD$ from noisy, incomplete distance data is by semidefinite programming. This hinges on the one-to-one equivalence between EDMs with embedding dimension $\embdim$ and centered Gram matrices of rank $\embdim$. Define the geometric centering matrix $\mJ$ as
\begin{equation*}
    \mJ_N \bydef \mI - \frac{1}{N} \vone \vone^\T.
\end{equation*}
Then $\EDMgram(\mG)$ is an invertible map on the set of Gram matrices which correspond to centered point sets (implying $\mG \vone = \vzero$) with the inverse given by
\begin{equation*}
    -\frac{1}{2} \mJ_N \EDMgram(\mG) \mJ_N  = \mG.
\end{equation*}
In particular, we have the following equivalence that holds for matrices $\mD$ with a zero diagonal:
\begin{equation}
\label{eq:lower_dim_characterization}
\begin{rcases*}
    \mD = \mathcal{D}(\mX) \\
    \affdim(\mX)  \leq \embdim
\end{rcases*} \quad
\Longleftrightarrow \quad
\begin{cases}
    -\frac{1}{2} \mJ_N \mD \mJ_N \succeq 0 \\
    \rank(\mJ_N \mD \mJ_N) \leq \embdim.
\end{cases}
\end{equation}
where $\mathcal{D}(\mX) = \mathcal{K}(\mX^{\T}\mX)$ and $\affdim$ denotes the dimension of the smallest affine space that can hold $\mX$. In other words, instead of directly searching for the points $\mX$ given distance data, we can search for the suitable Gram matrix. 

\begin{figure*}[h]
\centering
\includegraphics[width=1 \linewidth]{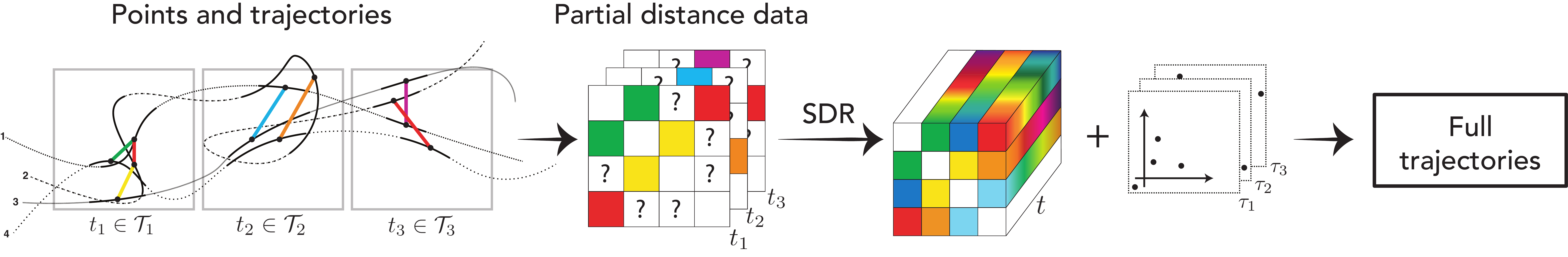}
\caption{KDGP: Estimate an embedded trajectory for a given sequence of partial pair-wise distances at different times, $t_1, \cdots, t_T$. We estimate the corresponding KEDM with a semidefinite relaxation \cref{alg:sdp}, and then use anchors to estimate the trajectories.}
\label{fig:kdgp-kedm-illustration}
\end{figure*}

Let $\wt{\mD}$ be the noisy, incomplete EDM from which we want to estimate the point locations, with unknown entries replaced by zeroes. Define  the mask matrix $\mW = (w_{ij})$ as
\begin{equation*}
    w_{ij} \bydef
    \begin{cases}
       1, \ (i, j) \in E \\
       0, \ \text{otherwise}
    \end{cases}.
\end{equation*}
This mask matrix will let us compute the loss only on those entries that were actually measured. Note that $\mW$ is precisely the adjacency matrix of the undirected graph from Problem \ref{prob:dgp}.

Then the above discussion is summarized in the following rank-constrained semidefinite program:
\begin{align} 
& \underset{\mG}{\text{minimize}}       & &  \| \wt{\mD} -\mW \circ \mathcal{K}\left(\mG \right) \|_{F}^{2} & \label{eq:stat_dgp} \\
& \text{subject to}     && \mG \succeq \mat{0} & \nonumber \\
&  & & \mG\vone = \vzero \nonumber \\
&                   & &   \rank \left( \mG \right) \leq \embdim. \nonumber
\end{align}

Since the Gram matrix (Gramian) is linearly related to the EDM, the objective function is convex. However, the rank constraint, $\rank \left( \mG \right) \leq \embdim$, makes the feasible set in \eqref{eq:stat_dgp} non-convex. Note that \eqref{eq:dgp_in_x} is also a non-convex program. However, while formalization in the $\mX$ domain gives no obvious way to convexify the problem, in the $\mG$ domain we can achieve convexity by discarding the rank constraint. It has been repeately shown in the literature that this \emph{semidefinite relaxation} achieves very good results \cite{Dokmanic:2015eg,Krislock:2012xx}. An intuitive explanation is that while the rank condition ensures the correct embedding dimension, the constraint that $\mG$ be positive semidefinite (in other words, that it be a Gramian) enforces a number of geometric constraints. For instance, it ensures that the entries of the EDM verify triangle inequalities, but also many more subtle properties beyond this. As long as the number of measured distances is large enough,\footnote{An empirical study of what ``large enough'' means is available in \cite{Dokmanic:2015eg}} we can expect these constraints to ensure the right embedding dimension.

The constraint $\mG \vone = \vzero$ serves to set the centroid of the recovered point set at the origin of the coordinate system as it implies $\mX \vone = \vzero$. This resolves the translational invariance of the problem. The remaining rotational (and reflection) invariance must be resolved once the points are estimated from the Gramian. The Gramian itself is of course invariant to the rotations of the point set since $\mG = \mX^\T \mX = (\mU \mX)^\T \mU \mX$ for any orthonormal matrix $\mU \in \R^{d \times d}$.

\subsection{Orthogonal Procrustes Problem}
\label{sub:procrustes}

As mentioned before, the EDM only specifies the point set up to a rigid transformation (rotation, translation, and reflection). If the task requires determining absolute positions of points, the standard method is to designate a subset of points as \emph{anchors} whose positions are known, and use anchors to align the reconstructed point set.

Let $\mX_a \in \R^{\embdim \times N_a}$ be the submatrix (a selection of columns) of $\mX$ that should be aligned with the anchors listed as columns of $\mY \in \R^{\embdim \times N_a}$. We note that the number of anchors---columns in $\mX_a$---is typically small compared with the total number of points---columns in $\mX$.

We first center the columns of $\mY$ and $\mX_a$ by subtracting the corresponding column means $\vy_c = \mY \mJ_{N_a}$ and $\vx_{a,c} = \mX_a \mJ_{N_a}$, obtaining matrices $\overline{\mY}$ and $\overline{\mX}_a$.  Next, we perform the orthogonal Procrustes analysis---we search for the rotation and reflection that best maps $\overline{\mX}_a$ onto $\overline{\mY}$:
\begin{equation}
    \label{eq:orthogonal_procrustes}
    \mR = \argmin_{\mQ:\mQ \mQ^\T = \mI} \norm{\mQ \overline{\mX}_a - \overline{\mY}}_F^2.
\end{equation}
The solution to \eqref{eq:orthogonal_procrustes} is given by the singular value decomposition (SVD) \cite{Schonemann:1964tj} as follows.
Let $\mU \mSigma \mV^\T$ be the SVD of $\overline{\mX}_a \overline{\mY}^\T$; then $\mR = \mV \mU^\T$. 
The best alignement is applied to the reconstructed point set as
\begin{equation*}
\mX_{\text{aligned}} = \mR (\mX - \vx_{a,c} \vone^\T) + \vy_c \vone^\T.
\end{equation*}

\section{Trajectory Models and Basis Gramians}
\label{sec:trajectories_and_basis_gramians}

In order to extend the EDM-based tools to the KDGP, we must define the class of trajectories $\mathcal{X}$. We introduce two trajectory models---polynomial and bandlimited---and show how they can be parameterized in terms of the so-called basis Gramians. 

The chosen trajectory models are standard; they model many interesting trajectories. The polynomial model is common in simultaneous localization and mapping as well as tracking, where it appears as constant velocity or constant acceleration model \cite{Wang:2003gd,wang:2003ladar}. The bandlimited model describes periodic trajectories of varying degrees of smoothness which are locally well-approximated by polynomials.

We use similar notation as in the static case. Let $\mX(t) = [\vx_1(t), \ldots, \vx_N(t)]$ be the trajectory matrix of $N$ points in $\mathbb{R}^\embdim$ where $\vx_n(t)$ is the position of $n$-th point at time $t$. We define the KEDM in a natural way:
\begin{definition}[KEDM]
Given a set of trajectories $\mX(t) \in \R^{\embdim \times N}$, the corresponding KEDM is the time-dependent matrix $\mD(t) \in \R^{N \times N}[t]$ of time-varying squared distances between the points:
\begin{align} 
\mD(t) \bydef \mathcal{D}(\mX(t)).
\label{eq:GtoEDM}
\end{align}
\end{definition}

\subsection{Polynomial Trajectories} 
For a set of $N$ points in $\R^\embdim$, we define the set of polynomial trajectories of degree $P$ as
\begin{equation} \label{eq:x_poly_tra}
    \mathcal{X}_\text{poly} = \set{\sum_{p=0}^{P}{t^p \mA_p} \ \bigg| \ \mA_{p} \in \mathbb{R}^{\embdim \times N}, \ p \in \set{0, \ldots, P}}.
\end{equation}
For $\mX(t) \in \mathcal{X}_\text{poly}$, the Gramian at time $t$ can be written as
\begin{equation}\label{eq:g_poly_tra}
    \mG(t) =  \sum_{k=0}^{K}{\mB_k t^k},
\end{equation}
where $\mB_k = \sum_{i=0}^{k}{\mA_i^\T\mA_{k-i}}$ and $K = 2P$. Similar to the static case, our goal is to cast the trajectory retrieval problem as a semidefinite program; we do so via the time-dependent Gramian in  \Cref{sec:kinetic_gramian_estimation}. 

The key step is to parameterize the problem entirely in terms of (constant) positive semidefinite matrices, instead of the parameterization in terms of $\mA_p$ or $\mB_k$. To do so, we fix $K + 1$ time instants $\tau_0, \ldots, \tau_{K}$ and define $\mG_k \bydef \mG(\tau_k)$. The matrices $\mG_k$ should be interpreted as \emph{elementary}, or \emph{basis} Gramians in the sense that the Gramian $\mG(t)$ can be written as a linear combination of $\mG_0,\ \ldots,\mG_{K}$ as elaborated in the following proposition.

\begin{proposition}\label{prop:polyG}
Consider the polynomial trajectory in \eqref{eq:x_poly_tra}. Let $\mG_k, \ k \in \set{0, 1, \ldots, K}, \ K = 2P$ be given as above with $\tau_k$ all distinct. Then we have 
\begin{equation*} 
  \mG(t) =  \sum_{k=0}^K w_k (t) \mG_k, 
\end{equation*}
with the weights $\vw(t) = [w_0(t), \cdots, w_K(t)]^\T$ given as
\begin{equation*}
\vw(t) =
\begin{pmatrix}
1 & 1 & \cdots & 1\\ 
\vdots & \vdots &  \ddots & \vdots \\ 
\tau_{0}^{K} & \tau_{1}^{K} & \cdots & \tau_{K}^{K} 
\end{pmatrix}^{-1}            
\begin{pmatrix} 
  1 \\ t \\ \vdots \\ t^{K} 
\end{pmatrix}.
\end{equation*}
\end{proposition}

\begin{proof}
The Gramians can be written as linear combinations of a set of monomial terms (cf. \eqref{eq:g_poly_tra}), which gives 
\begin{equation}
\label{eq:basis_system}
\begin{aligned}
  \mG_0  \ & = \  \mB_0  + \tau_0 \mB_1 + \cdots +  \tau_0^{K} \mB_{K} \\
  & \vdots \\
  \mG_{K}  \ & = \  \mB_0  + \tau_{K}\mB_1 + \cdots +  \tau_{K}^{K} \mB_{K}.
\end{aligned}
\end{equation}
Each matrix equation in \eqref{eq:basis_system} consists of $N \times N$ scalar equations for entries of $\mG_k$. Focusing on a particular entry $(i, j)$ gives a usual linear system $\vg = \mM \vb$ with column vector $\vg = [g_0,\cdots,g_K]^\T$ where $g_k$ is $(i, j)$-th element of $\mG_k$, the matrix $\mM = \bydef [ \tau_{k}^{k'}]_{0 \leq k,k' \leq K}$, and $\vb = [b_0,\cdots,b_{K}]^\T $ where $b_k$ is $(i, j)$-th element of $\mB_k$. We also have from \eqref{eq:g_poly_tra} that $[\mG(t)]_{ij} = (1, t, t^2, \ldots, t^K ) \vb \bydef \vt ^\T \vb$. Since $\tau_k$ are distinct, the square Vandermonde matrix $\mM$ is invertible. We have $\vb = \mM^{-1} \vg$, which gives $[\mG(t)]_{ij}  =  \vt^\T \mM^{-1} \vg$. Denoting $\vw(t) = (\mM^\T)^{-1} \vt$ we have that $[\mG(t)]_{ij} = \vw(t)^\T \vg = \sum_{k=0}^K w_k(t) [\mG_k]_{ij}$ which proves the claim.
\end{proof}

This result is a matrix version of Lagrange interpolation. Since entries of $\mG(t)$ are polynomials of degree $2P$ in $t$, they are completely determined by their values at $2P+1$ points. However, in this Gram matrix version it gives us something rather useful: a way to write a positive semidefinite $\mG(t)$ as a linear combination of positive semidefinite $\mG_k$, which lends itself nicely to convex optimization. We note that the question of how to choose the sampling times $\tau_k$ is beyond the scope of this article, though we give empirical results in Section \ref{sec:results}.

\subsection{Bandlimited Trajectories}
The second model we consider are periodic bandlimited trajectories. For a set of $N$ points in $\R^{\embdim}$, the set of periodic bandlimited trajectory of degree $P$ can be written as
\begin{equation}\label{eq:x_sine_tra}
    \begin{aligned}
    &\mathcal{X}_\text{BL} = 
     \bigg\{ \mB_0 + \sum_{p=1}^{P}\{\mA_p \sin(p \omega t) + \mB_p \cos(p \omega t)\ \bigg| \\ 
    & \ \ \ \mA_p, \mB_0, \mB_p \in \R^{\embdim \times N}, p \in \set{1, \ldots, P}, \omega \in \R^+ \bigg\}.
    \end{aligned}
\end{equation}
\begin{figure}[t]
    \centering
    \includegraphics[width=.9\linewidth]{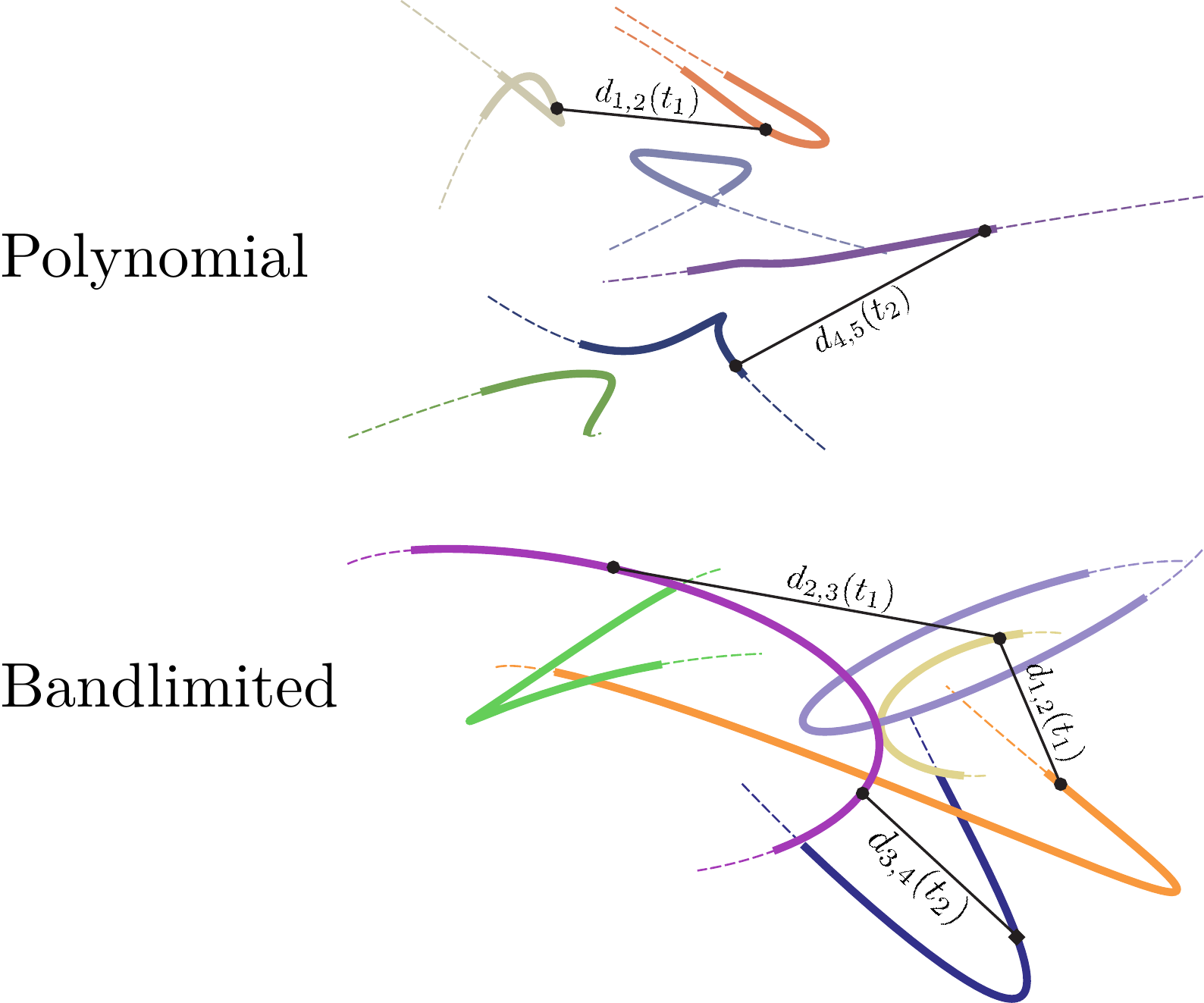}
    \caption{Example polynomial  ($P = 5$) and bandlimited ($P = 3, \omega = \pi / 4$) trajectories. Solid parts correspond to the sampling window $[T_1, T_2] = [-0.8,0.8]$.  The goal is to retrieve the trajectories of all points by measuring $d_{i,j}(t) = \| \vx_i(t)-\vx_j(t)\|_{2}$ at different times and for different $(i,j)$-pairs.}
\end{figure}
Similar to the polynomial case, we represent the Gramian $\mG(t)$ as a linear combination of some Gramian basis.
\begin{proposition}\label{prop:sineG}
Consider the bandlimited trajectory in \eqref{eq:x_sine_tra}. Let $\mG_k, \ k \in \set{0, 1, \ldots, K}, \ K = 4P$ be given as above with $\tau_k$ all distinct (modulo $\frac{2\pi}{\omega}$). We have 
\begin{equation*}
\mG(t) =  \sum_{k=0}^K w_k (t) \mG_k 
\end{equation*}
with the weights $\vw(t) = [w_0(t), \cdots, w_K(t)]^\T$ given as
\begin{align*}
\vw(t) =
\begin{pmatrix}
1 & \cdots & 1\\ 
\sin(\omega \tau_0) &   \cdots & \sin(\omega \tau_K) \\ 
\cos(\omega \tau_0) &  \cdots & \cos(\omega \tau_K) \\ 
\vdots  &  \ddots & \vdots \\ 
\cos(2P\omega \tau_{0}) & \cdots & \cos(2P\omega \tau_{K})
\end{pmatrix}^{-1}            
\begin{pmatrix}
  1 \\ \sin(\omega t) \\ \cos(\omega t)\\ \vdots \\ \sin(2P \omega t) \\ \cos(2P \omega t)
\end{pmatrix}.
\end{align*}
\end{proposition}
\begin{proof} 
The proof is analogous to the polynomial case. We only need to show that the system matrix is full rank which is a standard result \cite{Bass:2005dr}.
\end{proof}
We have thus developed a way to write a time-dependent Gramian in terms of a linear combination of positive semidefinite (constant) basis Gramians.

\subsection{Ambiguities Beyond Rigid Transformations in KDGP}

Same as the static DGP, the KDGP suffers from rigid transformation ambiguity. Namely, the rotated and translated trajectory sets cannot be distinguished from pairwise distance data. However, since at every time instant we can apply a different rigid transform, the set of ambiguities that arise in the KDGP is much larger than just the rigid transforms.

\begin{figure}[t]
\centering
\includegraphics[width=1 \linewidth]{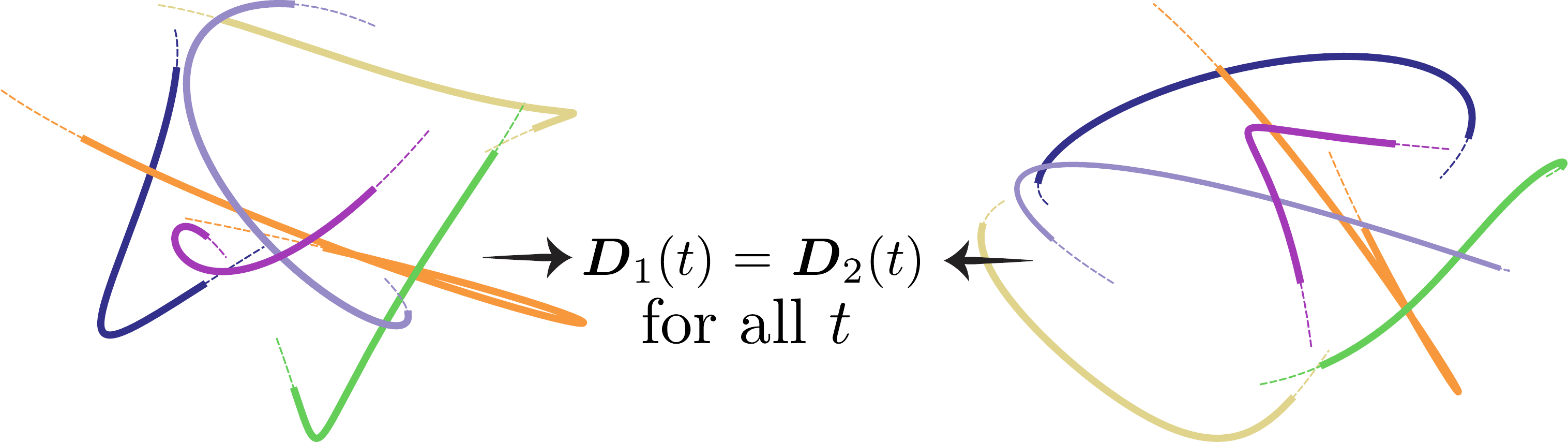}
\caption[Two numerical solutions]{Two trajectory sets which are not rigid transforms of each other, but which generate the same KEDM. Corresponding points have the same color.}\label{fig:kdgp-ambiguity}
\end{figure}

In particular, trajectory sets which look rather differently (nothing like rotations and translations of each other) could generate the same KEDM. We give an example in \Cref{fig:kdgp-ambiguity}. The following straightforward proposition characterizes trajectories that lead to the same KEDM. 
\begin{proposition} \label{prop:distAmb}
Let  $\mX(t)$, $\mY(t)$ be arbitrary trajectories in $\R^{d}$. Define distance equivalence relation as
\begin{equation*}
\mX \stackrel{\mathcal{D}}{\sim} \mY 
\end{equation*}
if and only if $\mathcal{D}(\mX(t)) = \mathcal{D}(\mY(t)), ~ \forall t$. Then, $\mX \stackrel{\mathcal{D}}{\sim} \mY $ if and only if $\mY(t) = \mU(t) \mX(t) + \vc(t) \vone^\T$ where $\mU(t)^\T \mU(t) = \mI$ and $\vc(t)$ is a $\embdim$-dimensional time-varying vector.
\end{proposition}

Requiring that the trajectories follow a particular model (for example polynomial or bandlimited) limits possible choices of the time-varying rigid transform parameters $\mU(t)$ and $\vc(t)$. In particular, known results on spectral factorization of polynomial matrices imply that the orthogonal $\mU(t)$ must be a constant matrix. On the other hand, as long as $\vc(t)$ is polynomial (or bandlimited) of the same degree as $\mX(t)$, it is a legal choice in the sense that the trajectories remain polynomial or bandlimited. But even with a fixed $\mU(t) = \mU$, varying $\vc(t)$ can produce trajectories of rather different shapes which are indistinguishable from their time-varying distances. 

In \Cref{sec:sf}, we propose a method for spectral factorization of kinetic Gramians based on anchor points and show how it resolves the described ambiguities.
In our algorithms we will choose $\vc(t)$ so that the centroid of the point set is kept fixed at the origin at all times, and then recover the correct centroid using anchor points. The following proposition will be useful.

\begin{proposition} \label{prop:AXG}
For trajectories of the form \eqref{eq:x_poly_tra} or \eqref{eq:x_sine_tra}, where $N \geq \embdim$, the following statements are equivalent:
\begin{enumerate}
\item All coefficient matrices have zero mean columns;
\item $\mX(t) \vone = \vzero, ~ \forall t \in \R$;
\item $\mG_k \vone = \vzero, ~ \forall k \in \{0, \cdots ,K\}$.
\end{enumerate}
\end{proposition}
\begin{proof}
We establish $(1) \Leftrightarrow (2)$ and $(2) \Leftrightarrow (3)$ for the polynomial trajectories and leave the straightforward extension to bandlimited trajectories to the reader. It is obvious that $(1)$ implies $(2)$ and $(2)$ implies $(3)$.

\noindent \textbf{(2) implies (1):} We have $\mX(t)\vone = \sum_{p=0}^{P} (\mA_p\vone) t^{p} = \vzero$. Since the monomials $\set{t \mapsto t^p}_{p=0}^P$ form a linearly independent set, the coefficients $\mA_p \vone $ must all be zero, or in other words, the column centroid of all $\mA_p$ must be at the origin.

\noindent \noindent \textbf{(3) implies (2):} Since $\mG(t) \vone =  \sum_{k=0}^{K}{w_k(t)\mG_k \vone} = \vzero$, we have 
\(
\mG(t)\vone = \mX(t)^\T\mX(t) \vone = \vzero.
\)
The polynomial matrix $\mX(t) \in \R^{\embdim \times N}[t]$ of degree $P$ with $\max_{t \in \R} \rank{\mX(t)} = \embdim$ is rank deficient at finitely many times. To see this note that $\det \mX(t) \mX(t)^\T$ is a polynomial of degree $2P\embdim$ which is not identically zero since $\max_{t \in \R} \rank{\mX(t)} = \embdim$. Therefore, it has at most $2P\embdim$ real roots. Thus, $\mX(t)^\T$ is a tall and full rank matrix except at an at most a finite number of times. Therefore, we have $\mX(t) \vone = \vzero ~ \text{for all}~t \in \R$. 
\end{proof}

\section{Computing the KEDM from Noisy, Incomplete Data by Semidefinite Programming} 
\label{sec:kinetic_gramian_estimation}

In this section we use the basis Gramian representation to derive an algorithm that solves the KDGP. Just as in the static case, we can either search directly for the set of trajectores $\mX(t)$ which reproduces the measured distances, or we can search for the time-varying Gramian $\mG(t)$ and use spectral factorization to estimate $\mX(t)$. In the static case, the two formulations are equivalent (they produce the same solution up to a rigid transform), but the formulation in terms of the Gram matrix led to a convenient semidefinite relaxation. In the time-varying case, we again state both formulations, and argue that the difference is now more significant.

To treat polynomial and bandlimited trajectories at once, we define the symbol $\mTheta$ to mean either $\mTheta = \{ \mA_p \}_{p=0}^{P}$ for the polynomial model or $\mTheta = \{ \mA_p, \mB_p \}_{p=1}^{P} \cup \set{\mB_0}$ for the bandlimited model, and similarly let $\mX_{\mTheta}(t) = \sum_{p = 0}^P \mA_p t^p$ (resp. $\mX_{\mTheta}(t) = \mB_0 + \sum_{p = 1}^P \mA_p \cos(p \omega t) + \mB_p \sin(p\omega t)$).

\paragraph{Formalizing in $\mX$ domain}

In this case, trajectory retrieval is written as
\begin{equation}\label{eq:opt_X}
    \begin{aligned}
    & \underset{\mTheta \in \mathcal{A}}{\text{minimize}}       & &  \sum_{i=1}^T{ \alpha_i \norm{ \wt{\mD}_{t_i} - \mW_i \circ \mathcal{D} \left(\mX_{\mTheta}(t_i) \right) }_{F}^{2}} &  \\
    & \text{subject to}     && \mX_{\mTheta}(t) \vone = \vzero, \forall \, t \in \R, 
    \end{aligned}
% \min_{ \Theta \in \mathcal{A} } \sum_{i=1}^\T{ \alpha_i \| \tilde{\mD}_{t_i} - \mW_i \circ \mathcal{D} \left(\mX_{\Theta}(t_i) \right) \|_{F}^{2}}
\end{equation} where $\mathcal{D}(\mX) = \mathcal{K}(\mX^\T \mX)$, $\wt{\mD}_{t_i}$ is the matrix of partial measured distances at time $t_i$, $\mW_i$ is the adjacency matrix corresponding to measurements, $\alpha_i \geq 0$ are non-negative weights, and $\mathcal{A}$ is the set of all feasible parameters. It is not hard to see that the objective in \eqref{eq:opt_X} is nonconvex in $\mTheta$ (even though the constraint set is convex by \Cref{prop:AXG}). 
Hence, this problem involves minimizing a nonconvex functional which is in general difficult.

\paragraph{Formalizing in $\mG$ Domain}

Next, we derive a semidefinite program inspired by \eqref{eq:stat_dgp} for the trajectory recovery problem. The key ingredient is the basis Gramian representation of $\mG(t)$ from \Cref{sec:trajectories_and_basis_gramians}. Since the actual kinetic Gramian is linear in basis Gramians, the overall objective will be convex as long as the data fidelity metric is convex. The latter holds true since we use the squared Frobenius norm:
\begin{equation}\label{eq:opt_G}
    \begin{aligned}
    & \underset{(\mG_k : \mG_k \succeq 0)_{k=0}^K}{\text{minimize}}       & &  \sum_{i=1}^\T{ \alpha_i \norm{ \wt{\mD}_{t_i} -\mW_i \circ \mathcal{K}\left(\sum_{k=0}^K w_k(t_i) \mG_k\right) }_{F}^{2}} &  \\
    & \text{subject to}     && \mG(t) \vone = \vzero, \forall t \in 
    \R \\
    &&& \mG(t) \succeq 0, \forall t \in \R \\\
    &&& \max_{\mathclap{t \in \R}}~\rank \mG(t) = \embdim.
    \end{aligned}
\end{equation}
The constraints ensure that the solution corresponds to a time-varying Gramian $\mG(t)$ with correct rank. We note that there is no rotation ambiguity associated with this formulation because the Gramian is invariant to rotation and reflection of the points. Translation ambiguity has been resolved by requiring that $\mG(t) \vone = \vzero$ which implies that the recovered point set shall be centered at all times.

\subsection{Equivalence Between \eqref{eq:opt_X} and \eqref{eq:opt_G}}

The two formulations are equivalent if for every possible set of measurements, the solution sets produce the same KEDM. Denoting the optimizers (which could be sets) by $\mTheta^*$ and $(\mG_k^*)_{k = 0}^K$, it should hold that
\begin{equation*}
\mathcal{D} \left(\mX_{\Theta^*}(t) \right) =  \mathcal{K}\left(\sum_{k=0}^K w_k(t) \mG^{*}_k\right), ~ t \in \R.
\end{equation*}

By Propositions \ref{prop:polyG} and \ref{prop:sineG}, for any optimizer $\mTheta^*$ of \eqref{eq:opt_X} and the corresponding trajectory $\mX_{\Theta^*}(t)$, we can find a Gramian basis $(\wt{\mG}_k)_{k=0}^{K}$ such that $\mathcal{D} \left(\mX_{\Theta^*}(t) \right) =  \mathcal{K}\left(\sum_{k=0}^K w_k(t) \wt{\mG}_k\right)$. Therefore,
\[
    J_1(\mTheta^*) = J_2((\wt{\mG}_k)_{k=0}^K) \geq J_2((\mG_k^*)_{k=0}^K),
\]
where $J_1$ denotes the loss in \eqref{eq:opt_X}, and $J_2$ denotes the loss in \eqref{eq:opt_G}.
The question is whether this inequality can be made strict. Could the solution to \eqref{eq:opt_G} lead to a Gramian $\mG(t)$ with no corresponding trajectory in $\mathcal{A}$? An in-depth study of this question is beyond the scope of this paper, but to see that this is indeed possible consider a contrived case of no measurements at all, that is to say, a feasibility search.

Trivially, any $\mTheta \in \mathcal{A}$ is a solution to \eqref{eq:opt_X} and any set of basis Gramians $(\mG_{k} : \mG_k \succeq \mat{0})_{k=0}^{K} $ is a solution to \eqref{eq:opt_G}. By Lemma \ref{lem:specfactor} every Gramian $\mG(t)$ produced by its basis $(\mG_{k})_{k=0}^{K}$ has a polynomial spectral factor, that is, it corresponds to a polynomial trajectory $\mX(t)$ such that $\mG(t) = \mX(t)^\T \mX(t)$. Even though $\mG(t)$ is real, its spectral factor, however, need not be; see \cite{tabaghi2019real} for a characterization of rank-deficient polynomial Gramians $\mG(t)$ without real spectral factors. This situation is fundamentally different from what we had in the static case. Hence, we can construct feasible ``complex trajectories'' which are outside of $\mathcal{A}$. Consequently, the constraints in \eqref{eq:opt_X} are necessary but not sufficient for the two formulations to be equivalent.  Nonetheless, they become equivalent with sufficient measurements:

\begin{proposition}\label{thm:main}
Suppose that $\wt{\mD}_i = \mW_i \circ \mathcal{D}(\mX_\Theta(t_i))$ and \eqref{eq:opt_G} has a unique optimizer $\mG^*(t) = \sum_{k = 0}^K w_k(t) \mG_k^*$. Then 
\begin{equation}
    \label{eq:minimizer-is-exact}
    \mG^*(t) = \mX_{\mTheta^*}(t)^\T \mX_{\mTheta^*}(t)
\end{equation}
where $\mX_{\Theta^*}(t)= \mU \mX_{\mTheta}(t)\mJ_N$ for some orthogonal matrix $\mU \in \R^{\embdim \times \embdim}$ (that is, it is a centered, rotated version of the true geometry).
\end{proposition}

\begin{proof} 
We prove this proposition by construction. Let us define $\mG^*_k = \mJ_{N}\mX_{\Theta}(\tau_k)^\T\mX_{\Theta}(\tau_k)\mJ_{N}$ for $k \in \{0, \cdots, K\}$ and $\mG^*(t) = \sum_{k=0}^{K}{w_k(t)\mG^*_k}$. From \Cref{prop:polyG,prop:sineG}, we deduce that $\mG^*(t) = \mJ_{N}\mX_{\Theta}(t)^\T\mX_{\Theta}(t)\mJ_{N}$. Hence, $\mG^*(t)$ belongs to the feasible set of \eqref{eq:opt_G} as $\mJ_{N}\mX_{\Theta}(t)^\T\mX_{\Theta}(t)\mJ_{N}$ is a zero-mean positive semidefinite matrix for all $t \in \R$, with rank at most $\embdim$. On the other hand, since $\mathcal{D}(\mX_\Theta(t)) = \mathcal{K}(\mG^*(t))$,  we have $J_2(\mG^*) = 0$. Finally, since \eqref{eq:opt_G} has a unique solution, the minimizer of \eqref{eq:opt_G} must have the form \eqref{eq:minimizer-is-exact}.
\end{proof}

It is useful to interpret the two approaches in \eqref{eq:opt_X} and \eqref{eq:opt_G} in terms of graph-based definition of the KDGP (Problem \ref{prob:kdgp}). The sequence of incomplete and noisy distances, $\wt{\mD}_{t_1}, \cdots, \wt{\mD}_{t_T}$ is modeled as a series of incomplete graphs whose edge weights correspond to the measured distances. The goal of KDGP is to find a node function $x(u,t)$ that maps vertices of measurement graphs to points in $\R^\embdim$ whose pairwise distances match the measured distances at sampling times $t_k \in \mathcal{T}$. From this perspective, the formulation \eqref{eq:opt_X} aims to directly estimate the node function $x(u, t)$ from distance measurements, while in formulation \eqref{eq:opt_G}, we break the KDGP into two subproblems: 

\begin{enumerate}
\item \textbf{Completing the measurement graphs:} This amounts to estimating the edge function, $f(e,t)$ for every $e \in E_t$ instead of only for $e \in E_i$, with $E_i$ being the edges measured at time $t_i \in \mathcal{T}$;
\item \textbf{Estimating the node function, $x(u,t)$}: This is equivalent to spectral factorization of the time-dependent Gramian.
\end{enumerate}
The formulation \eqref{eq:opt_G} solves the first subproblem since it outputs a time-varying Gramian $\mG(t)$ from which we easily get the KEDM as $\mathcal{K}(\mG(t))$. The second problem is addressed in Section \ref{sec:sf}.

Finally, we note that the KEDM formulation in \eqref{eq:opt_G} is a generalization of the static EDM formulation in \eqref{eq:stat_dgp}. To see the equivalence, note that static points are modeled by a polynomial of degree zero, $P=0$, in which case the Gramian becomes $\mG(t) = \mG_0$ since $w_0(t) = 1$. 

\subsection{Practical Considerations: Relax and Sample}

To get a practical algorithm for \eqref{eq:opt_G}, we sample the continuous-time semidefiniteness constraint,
\(
\mG(t) \succeq \vzero, ~ \forall t \in \R,
\)
and relax the non-convex rank constraint. In \Cref{alg:sdp}, we denote the set of sampling times for this constraint by $\mathcal{T}_\text{psd}$.

% If the rank of estimated Gramian $\mG(t)$ tends to deviate from $d$, typically at the times when we have no measurements, we can remedy this issue by picking a appropriate time, $\tau$, and minimize the following regularized objective function

% \begin{equation*}
% J_2((\mG_k))  + \lambda \| \sum_{k=0}^{K}{w_{k}(\tau) \mG_k} \|_{*}
% \end{equation*}
% where $J_2$ denotes the loss in \eqref{eq:opt_G} and $\| \cdot \|_{*}$ is the nuclear norm operator. This regularization will encourage a low rank estimated Gramian at time $\tau$.

In relaxations for static EDMs, instead of simply removing the rank constraint, it is often replaced by a regularizer. Perhaps counterintuitively (see \cite{Dokmanic:2015eg} for a longer discussion), a strategy that works well is to \emph{maximize} the rank of the Gram matrix, as this corresponds to pushing the points apart and minimizing the embedding dimension. We use a similar strategy in our KEDM semidefinite relaxation.

Even with this additional regularization, due to noise and numerical issues of the off-the-shelf semidefinite solvers, the recovered Gram matrices will rarely be \emph{exactly} rank-$\embdim$. To address this, we apply a standard rank projection to the retrieved Gramians by setting the least significant $N - \embdim$ singular values to 0.

\begin{algorithm}
\caption{Semidefinite relaxation for KEDM}\label{alg:sdp}
\begin{algorithmic}[1]
\Procedure{$\mathtt{SDR}$}{$\set{t_i}_{i=1}^T, \{\wt{\mD}_{t_i}\}_{i=1}^T, \set{\mW_i}_{i=1}^T$}
% \State Solve for $(\mG^{(i+1)}_{k})_{k=0}^{K}$, using any semidefinite algorithm with $\{\tilde{\mD}_{t_i}\}_{i=1}^T$ as input and warm start of $( \mG^{(i)}_{k})_{k=0}^{K}$. Trajectory model and all other parameters are assumed to be known.
\State Solve for $\set{\mG_k}$:
\begin{align*}
& \text{minimize}       & &  \sum_{i=1}^T \alpha_i \norm {\tilde{\mD}_{t_i} - \mW_i \circ \mathcal{K}\left(\sum_{k=0}^K w_k(t_i) \mG_k\right)}_{F}^{2} \\
&&& \quad \quad - \lambda \sum_{k=0}^K \trace(\mG_k)\\
& \text{w.r.t} & & \mG_0, \cdots, \mG_K  \succeq \mat{0} \\
& \text{such that} & &  \mG_k\vone = \vzero \qquad  \qquad \quad ~ k \in \{0, \cdots, K\},\\
&                   & &   \sum_{k=0}^{K}{w_k(t)\mG_k} \succeq \vzero   \qquad  t \in \mathcal{T}_{\text{psd}}.
\end{align*}
\State $\mG_k \gets \mathtt{RankProjection}(\mG_k, \embdim)$, $k \in  \set{0, \cdots, K}$
\State \textbf{return} $\wh{\mD}(t) = \mathcal{K} \left( \sum_{k=0}^{K} w_k(t) \mG_k \right) $
\EndProcedure
\end{algorithmic}
\end{algorithm}

% \begin{figure}[t]
% \centering
% \includegraphics[width=0.45\textwidth]{measurements.eps}
% \caption{The graph-theoretic interpretation of the proposed semidefinite program. The top row represents the set of incomplete measurements as a sequence of graphs. The bottom row represents the solution of the SDP which is a sequence of completed measurement graphs.}
% \end{figure}

\section{Spectral Factorization of the Gramian} \label{sec:sf}

\Cref{alg:sdp} produces a time-varying Gramian whose KEDM best represents the measured distance sequence. In this section, we show how to estimate the corresponding trajectory by factorizing the Gramian as $\mG(t) = \mX(t)^\T \mX(t)$ where $\mX(t)$ is the set of point trajectories. We know that the trajectory can only be estimated up to a time-invariant rotation (and possibly reflection) \cite{ephremidze2014elementary} and a time-varying translation. To resolve this uncertainty, we introduce a set of anchors---points whose absolute positions are known.

In practice, anchors might correspond to nodes that are equipped with a positioning technology such as GPS. Because the anchors move (unlike in the usual DGP), we have more possibilities for anchor measurements than in the static case. For our trajectory models, we only need to know the positions of the anchor points at some fixed, finite set of times, but we could measure positions of different sets of points at different times.

Given a spectral factor\footnote{One out of infinitely many possible.} $\overline{\mX}(t)$, of the time-varying Gramian, the true trajectory, $\mX(t)$, can be found as 
\begin{equation}\label{eq:XtoXbar}
\mX(t) = \mU \overline{\mX}(t) + \vx(t)\vone^\T + \mN(t).
\end{equation}
where $\mU$ is a $\embdim \times \embdim$ orthogonal matrix, $\vx(t)$ is a $\embdim \times 1$ time-varying vector and $\mN(t)$ represents the net effect of model mismatch and measurement noise. The matrix $\mU$ is constant (by the spectral factorization theorem) whereas the translation factor $\vx(t)$ is a function of time.
On the other hand, the translation factor $\vx(t)$ must belong to the same trajectory model as $\mX(t)$ (polynomial or bandlimited). Hence, $\vx(t)$ can be written as
\[
    \vx(t) = \mM \vz(t),
\]
where for the polynomial model, $\mathcal{X}_{\text{poly}}$ we have
\begin{equation*}
\vz(t) = [1, \ t, \ \cdots, \ t^P]^\T
\end{equation*}
and $\mM \in \R^{d \times (P+1)}$, and for the bandlimited model, $\mathcal{X}_{\text{BL}}$, 
\begin{equation*}
\vz(t) = [1, \sin(\omega t), \cos(\omega t), \cdots, \sin(P \omega t), \cos(P\omega t)]^\T
\end{equation*}
and $\mM \in \R^{ d \times (2P+1)}$.

%On the other hand, if $\overline{\mX}(t)$ belongs to $\mathcal{X}_{\text{poly}}$ or $\mathcal{X}_{\text{BL}}$, so should $\vx(t)$. We can thus write 
%\[
%    \vx(t) = \mM \vz(t),
%\]
%where for the polynomial model we have
%\begin{equation*}
%\vz(t) = [1, \ t, \ \cdots, \ t^P]^\T
%\end{equation*}
%and $\mM \in \R^{d \times (P+1)}$, and for the bandlimited model we have
%\begin{equation*}
%\vz(t) = [1, \sin(\omega t), \cos(\omega t), \cdots, \sin(P \omega t), \cos(P\omega t)]^\T
%\end{equation*}
%and $\mM \in \R^{ d \times (2P+1)}$.

A difficulty compared to the static case is that spectral factorization of polynomial Gram matrices is not straightforward and becomes brittle in the presence of noise. It is thus desirable to develop trajectory estimation methods that do not require full polynomial factorization. We show that this is possible at the expense of additional anchor measurements.

\subsection{Known Spectral Factor} 
\label{sub:known_spectral_factor}

We start by assuming that we have access to \emph{some} spectral factor $\overline{\mX}(t)$ such that $\mG(t) = \overline{\mX}(t)^\T \overline{\mX}(t)$. In this case, to estimate the unknown rotation and translation, we assume that at $L$ distinct times $\tau_1, \ldots, \tau_{L}$ we measure positions of points $\mathcal{I}_1, \ldots, \mathcal{I}_L$, with $\mathcal{I}_\ell$ being the index set of points whose positions are measured at $\tau_\ell$. We let $\mX_{\mathcal{I}_\ell}$ denote the column selection of $\mX(\tau_\ell)$ corresponding to indices in $\mathcal{I}_\ell$.

An estimate for $\mU$ and $\mM$ can be computed by solving
\begin{equation*}
\argmin_{\mU \in M_{\embdim}(\R), \mM \in \R^{\embdim \times L}} \sum_{\ell = 1}^L \norm{ \mX_{\mathcal{I}_\ell} - \mU \overline{\mX}(\tau_\ell) - \mM \vz(\tau_\ell)\vone^\T }_{F}^{2}
\end{equation*}
where $M_{\embdim}(\R)$ is the set of $\embdim \times \embdim$ orthonormal matrices and $L= P+1$ (resp. $2P+1$) for polynomial (resp. bandlimited) trajectories. This is a non-convex problem because $M_{\embdim}(\R)$ is a non-convex set. 

The above optimization can be decoupled as in standard Procrustes analysis provided that there exists a time $\wt{\tau}_\ell \in \set{\tau_1, \ldots, \tau_{L}}$ at which we know the positions of at least $\embdim + 1$ anchors. In this case, $\mU$ can be estimated at this time alone using the technique described in  Section \ref{sub:procrustes}. Once the rotation $\wh{\mU}$ is found, we can estimate the matrix $\mM$ by solving the following convex problem:
\begin{equation*}
\wh{\mM} = \argmin_{\mM \in \R^{\embdim \times L}} \sum_{\ell=1}^{L}{\norm{ \mM \vz(\tau_\ell) - \frac{1}{N_{\tau_\ell}} \big(\mX_{\mathcal{I}_\ell}(\tau_\ell) - \wh{\mU} \overline{\mX}(\tau_\ell)\big) \vone }_{2}^{2}}
\end{equation*}
where $N_{\tau_l} \geq 1$ for $\ell \geq 2$. 
Finally, we note that matching $\embdim$ instead of $\embdim + 1$ points would leave us with a flip ambiguity, so $\embdim + 1$ is indeed the smallest number of anchors that lets us use the Procrustes analysis.

\subsection{Unknown Spectral Factor (Practical Algorithm)} % (fold)
\label{sub:practical_algorithm_unknown_spectral_factor}

The previous section implies that $L+\embdim$ anchor points are necessary to estimate the rotation $\mU$ and translation $\mM$ provided that a spectral factor  $\overline{\mX}(t)$ of $\mG(t)$ is given. Unfortunately, algorithms for spectral factorization rely on unstable computations involving determinants and are often computationally demanding, which makes them unsuitable for our application where noise can be significant \cite{ephremidze2018algorithmization}. To avoid this step, we propose a method which relies on additional anchor measurements. 

Assume that at each of $L$ distinct times we measure positions of at least $\embdim+1$ anchors; as before, denote the anchor indices at time $\tau_\ell$ by $\mathcal{I}_\ell$, and the corresponding positions by $\mX_{\mathcal{I}_\ell}$. Now we can use Procrustes analysis at each time individually (that is, applied to constant matrices that are evaluations of time-varying matrices at these particular times) to estimate rotation and translation, $\wh{\mU}_{\tau_l}$ and $\wh{\vx}(\tau_l)$ at time $\tau_l$. Denote by $\overline{\mX}(\tau_\ell)$ any matrix such that $\overline{\mX}(\tau_\ell)^\T \overline{\mX}(\tau_\ell) = \mG(\tau_\ell)$; since this involves only constant matrices, we can use the eigendecomposition method described in Section \ref{sec:solving_dgp_edms} to compute $\overline{\mX}(\tau_\ell)$.

Note that in doing so, there is no guarantee that these ``marginal'' estimates for the rotation correspond to the unique global $\mU$ we are looking for, because we do not exploit any temporal model in computing the spectral factors $\overline{\mX}(\tau_\ell)$. In other words, all $\wh{\mU}_{\tau_l}$ could be distinct, and in principle they will. Nevertheless, we can use them to estimate the trajectory by solving the following problem:
\begin{equation}
\label{eq:spectral-factorization-convex}
\mTheta^* = \argmin_{\mTheta \in \mathcal{A}}~\sum_{l=1}^{L}{ \norm{ \mX_{\mTheta}(\tau_l) - \big( \wh{\mU}_{\tau_l}\overline{\mX}(\tau_l)+ \wh{\vx}(\tau_l) \vone^\T \big) }_{F}^{2}},
\end{equation}
and using $\mX_{\mTheta^*}(t)$ as our estimate of the trajectory. 

The logic behind \eqref{eq:spectral-factorization-convex} is that even though the matrices $\wh{\mU}_{\tau_\ell}$ are ``wrong'', the product $\wh{\mU}_{\tau_l}\overline{\mX}(\tau_l)$ is correct thanks to the anchors. With sufficiently many marginal estimates, there is a unique set of polynomial trajectories passing through them. The described procedure is summarized in Algorithm \ref{alg:spectral_factorization} and the complete KEDM trajectory localization algorithm with anchors in \Cref{alg:kedmAL}.

\begin{algorithm}
\caption{Spectral factorization}\label{alg:spectral_factorization}
\begin{algorithmic}[1]
\Procedure{$\mathtt{SpectralFactorization}$}{$\wh{\mD}(t)$, $\{\mX_{\mathcal{I}_\ell}\}_{l=1}^{L}$} \newline
\Comment{$\mX_{\mathcal{I}_\ell}$: Anchor points at different times, $\tau_1,\cdots, \tau_L$.}
\For{ $l\in \{1, \cdots , L\}$ }
\State $\wh{\mG}(\tau_l) \gets -\frac{1}{2}\mJ_N \wh{\mD}(\tau_l) \mJ_N$
\State $\overline{\mX}(\tau_l) \gets \wh{\mG}(\tau_l)^{1/2}$
\State Solve for $\wh{\mU}_{\tau_l}$ using Procrustes analysis
\State Estimate the translation at time $\tau_\ell$:
\begin{equation*}
\wh{\vx}(\tau_l) \gets \frac{1}{N_{\tau_l}}(\mX_{\mathcal{I}_\ell} - \wh{\mU}\overline{\mX}_{\mathcal{I}_\ell}(\tau_l)\mS_l)\vone
\end{equation*}
\State Estimate point positions at time $\tau_\ell$:
\begin{equation*}
\wh{\mX}(\tau_l)  \gets \wh{\mU}_{\tau_l}\overline{\mX}(\tau_l) + \wh{\vx}(\tau_l)\vone^\T
\end{equation*}
\EndFor
\State Find the trajectory:
\begin{equation*}
\Theta \gets \argmin_{\Theta \in \mathcal{A}} \sum_{\ell=1}^{L}{ \| \mX_{\Theta}(\tau_\ell) - \big( \wh{\mU}_{\tau_\ell}\overline{\mX}(\tau_\ell)+ \wh{\vx}(\tau_\ell) \vone^\T \big) \|_{2}^{2}}
\end{equation*}
\State \textbf{return} $\mX_{\Theta}(t)$
\EndProcedure
\end{algorithmic}
\end{algorithm}

\begin{algorithm}
\caption{End-to-end KEDM algorithm}\label{alg:kedmAL}
\begin{algorithmic}[1]
\Procedure{$\mathtt{KEDM}$}{$\{\mX_{\mathcal{I}_\ell}\}, \set{t_i}, \{\wt{\mD}_{t_i}\}, \set{\mW_i}$}
%\Comment{$\mX_0 \in \R^{d \times d}$: anchor points at $t=0$}
\State $\wh{\mD}(t) = \mathtt{SDR}(\set{t_i}, \{\wt{\mD}_{t_i}\}, \set{\mW_i})$
\State $\mX_\Theta(t) = \mathtt{SpectralFactorization}(\wh{\mD}(t)$, $\{\mX_{\mathcal{I}_\ell}\})$
\State \textbf{return} $\mX_{\Theta}(t)$
\EndProcedure
\end{algorithmic}
\end{algorithm}

%%%%%%%%%%%%%%%%%% END %%%%%%%%%%%%%%%%%%%%%%

\section{Simulation Results} \label{sec:results}

In this section we empirically evaluate different aspects of the proposed algorithm. We first study the influence of sampling time distribution in \Cref{sec:err_bound} as this choice affects the other experiments. In \Cref{sec:sparse}, we look at the maximum achievable measurement sparsity\footnote{We use the term ``sparsity'' to refer to  sparse or subsampled measured data, as is common in the inverse problems theory.}
\label{sec:sparse}: KDGP measurements are a sequence of incomplete EDMs and it is interesting to understand what proportion of missing entries we can tolerate.\footnote{In all experiments we sample the positive semidefinite constraint at random times. We have found empirically that this choice does not matter much, unlike the choice of measurement  times. The exact details can be found in the reproducible code at \url{https://github.com/swing-research/kedm/} which offers superior documentation.}
% \begin{itemize}
% \item For polynomial model, we uniformly generate samples $t_i$ from $[T^{-}, T^{+}]$ for some $T^{-} \ll 0 \ll T^{+}$. Then, $\mathcal{T}_{\text{psd}}^{+} = \set{e^{t_i}}$. Similarly $\mathcal{T}_{\text{psd}}^{-} = \set{-e^{t_i}}$ and $\mathcal{T}_{\text{psd}} = \mathcal{T}_{\text{psd}}^{+} \cup \mathcal{T}_{\text{psd}}^{-}$.
% \item For bandlimited, $\mathcal{T}_{\text{psd}}$ is comprised of uniformly generated samples in $[0,\frac{2\pi}{\omega}]$.
% \end{itemize}
%run a batch of experiments for various choices of trajectory parameters. 
Finally, in \Cref{sec:noise} we study the effect of measurement noise on the quality of the estimated trajectories.

We end this section by applying our algorithms to a synthetic problem of satellite localization from noisy and very sparse distance measurements.

%\textbf{How do you sample the psd constraint? Should mention.}

\subsection{Distribution of Sampling Times} \label{sec:err_bound}

The measurements in \Cref{alg:kedmAL} are a sequence of (incomplete) snapshots of KEDM at different times, $\{ \mW_i \circ \mathcal{D}(\mX(t_i)) \}_i$. We experiment with different choices of sampling times $\set{t_i}_i$. To exclude the influence of other factors, we assume having access to all pairwise distances, and we contaminate the measurements by noise. Note that without noise, we can compute the Gramian basis simply by solving a linear system of equations so that any sampling strategy with sufficiently many samples gives the perfect estimation. 

Let the true, noiseless distances be $d_{ij}(t) = \norm{\vx_i(t) - \vx_j(t)}$, and noisy measurements given as
\begin{align}
\wt{d}_{ij}(t) &=  d_{ij}(t) + n_{ij}(t)  \label{eq:measurement_noise}
\end{align}
where $n_{ij}(t) \sim \mathcal{N}(0, \sigma^2)$ is iid measurement noise. The correspponding KEDMs are $\mD(t) = [d_{ij}^{\,2}(t)]_{ij}$ and $\wt{\mD}(t) = [\wt{d}_{ij}^{\,2}(t)]_{ij}$.

To compare the different sampling protocols, we average the reconstruction error over many trajectory and noise realizations. The reconstruction error is defined as
\begin{equation*}
e_D(t) = \frac{\| \mD(t) - \wh{\mD}(t) \|_{F}}{\| \mD(t) \|_{F}}.
\end{equation*}
where $\wh{\mD}(t) = \mathtt{SDR}((\set{t_i}, \{\wt{\mD}_{t_i}\}, \set{\mW_i})\}_{i=1}^T)$ is the KEDM estimated by \Cref{alg:kedmAL}. The goal is to determine which sampling pattern  minimizes $e_D(t)$ for all $t$ in the interval of interest $[T_1, T_2]$.
\begin{figure}[t]
\centering
\includegraphics[width=1\linewidth]{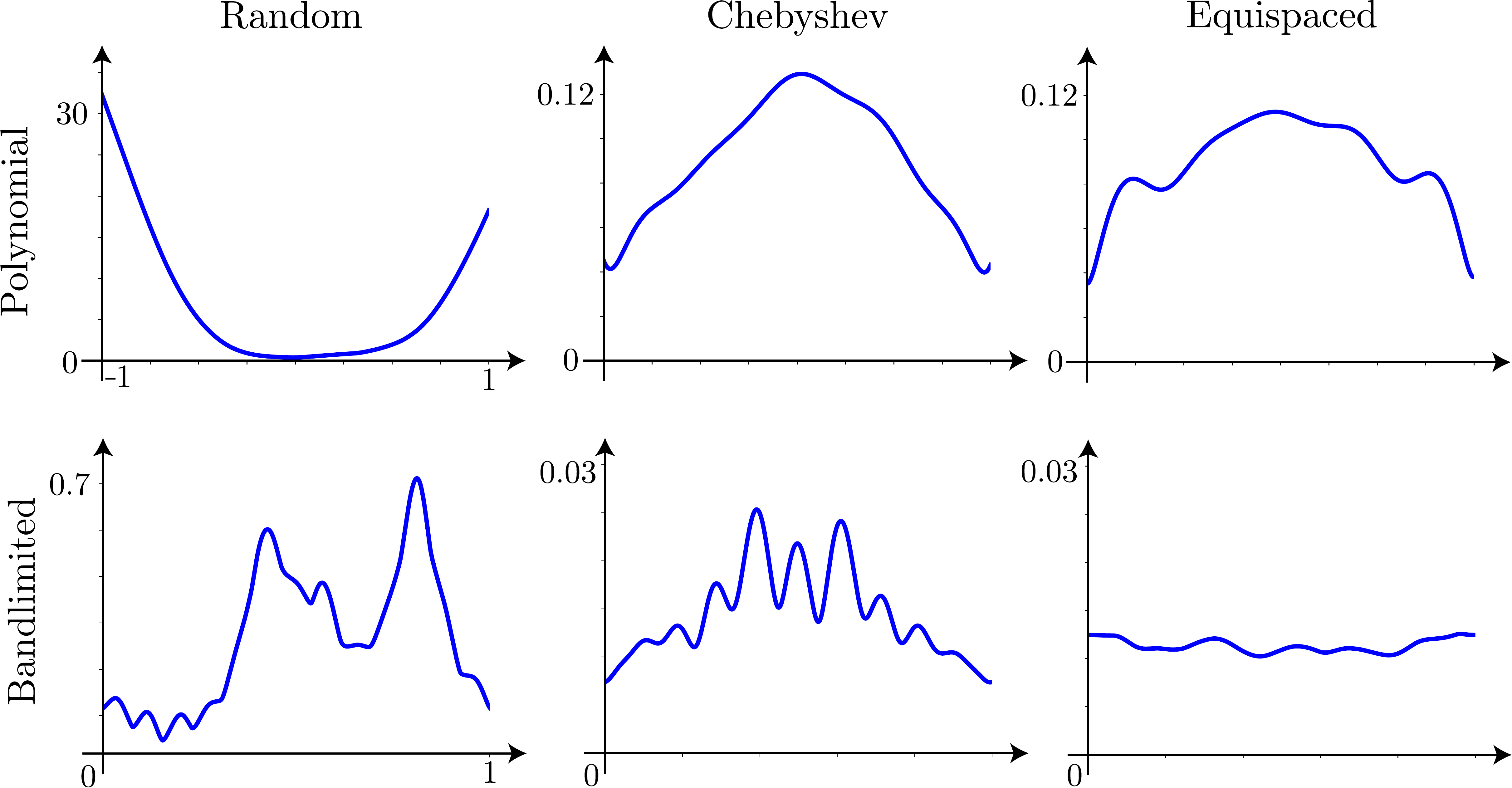}
\caption{Relative reconstruction error $e_D(t)$ averaged over $M=200$ realizations. The number of points is $N=10$, ambient dimension $d=2$, trajectory degree $P = 3$ and noise variance $\sigma^2 = 1$ for both models. The trajectory parameters, $\mA_p$, are drawn iid Gaussian---real valued for polynomial and complex for bandlimited with complex exponential basis. The sampled interval of interest is $[-1,1]$ for the polynomial and $[0, 1]$ for the bandlimited model.}
\label{fig:error_plots}
\end{figure}
In \Cref{fig:error_plots}, we show the average errors for the following sampling patterns:
\begin{itemize}
\item random: $t_i \sim \text{Unif}([T_1,T_2])$,
\item Chebyshev: 
\(
t_i  =  \tfrac{1}{2}(T_1+T_2) + \tfrac{1}{2}(T_2-T_1)\cos(\tfrac{2i-1}{2T}\pi),
\)
\item equispacesd: $t_i = T_1 +(T_2-T_1)	\frac{i}{T}$,
\end{itemize}
where $i = 1,\cdots, T$. We can see that random sampling performs poorly for both the polynomial and the bandlimited model. Chebyshev and equispaced nodes give a similar relative error, with equispaced nodes performing slightly better for the bandlimited model. Studying individual realizations shows that the worst-case error for Chebyshev and equispaced sampling is on the same order as the average error, but it is much worse for random sampling: large reconstruction errors occur when two consecutive measurement times are far apart. In the following experiments, we use equispaced measurement times.

\subsection{Measurement Sparsity}

Trajectory estimation from distances is a nonlinear sampling problem, with trajectory models allowing us to trade spatial for temporal samples.
Here we empirically study the maximum sparsity level for spatial measurements. Given a sequence of measurement masks $\mW_1, \cdots, \mW_T \in \{0,1\}^{N \times N}$, the sparsity level, $0 \leq S \leq 1$, is defined as the ratio of average to total number of pairwise distances:
\begin{equation*}
S = \frac{1}{\binom{N}{2}}\frac{1}{T}\sum_{i=1}^{T}{\textrm{\# of missing measurements at time}~t_i}.
\end{equation*}
We can expect the maximum sparsity level to vary with factors such as 
the trajectory model, temporal sampling pattern, measurement masks, and noise.
To evaluate it, we fix parameters the trajectory class, degree, number of points, and ambient dimension. We declare a localization experiment successful if the relative trajectory mismatch,
\[
e_X = \int_{\mathcal{T}}{{\| \mX(t)-\wh{\mX}(t)\|_{F}}/{\| \mX(t)\|_{F}} \, \di t},
\]
which we approximate by discretizing $\mathcal{T}$, is below some prescribe threshold $\delta$. We are interested in numerically evaluating the probability that the localization succeeds (within tolerance $\delta$) if on average over sampling times, $m$ pairwise distances are missing. Denote this probability by $p(\delta, m)$. We would like to find conditions on $m$ such that $p(\delta, m)$ is large. In particular, for $0 \leq q < 1$, let $\overline{m}(\delta, q)$ be the largest $m$ such that $p(\delta, m) \geq q$.

We run $M$ localization trials for different realizations of random trajectories, and denote the number of succesful trials by $M_1$. For a given average number of missing pairwise distances $m$, the probability of correct localization is estimated as $\wh{p}_M(\delta, m) = \frac{M_1}{M}$. The estimate of $\overline{m}(\delta, q)$ is then simply
\[
    \wh{\overline{m}}(\delta, q) := \max \set{m \ : \ \wh{p}_M(\delta, m) \geq q}.
\]
To compute $\wh{\overline{m}}(\delta, q)$, we increase the number of missing measurements per sampling time, $m$, and count the number of $\delta$-accurate estimates to compute $\wh{\overline{m}}(\delta, q)$ and the corresponding $\wh{S}(\delta, q) = {\wh{\overline{m}}(\delta, q)}/{\binom{N}{2}}$.

\begin{figure}[t]
\centering
\includegraphics[width=1\linewidth]{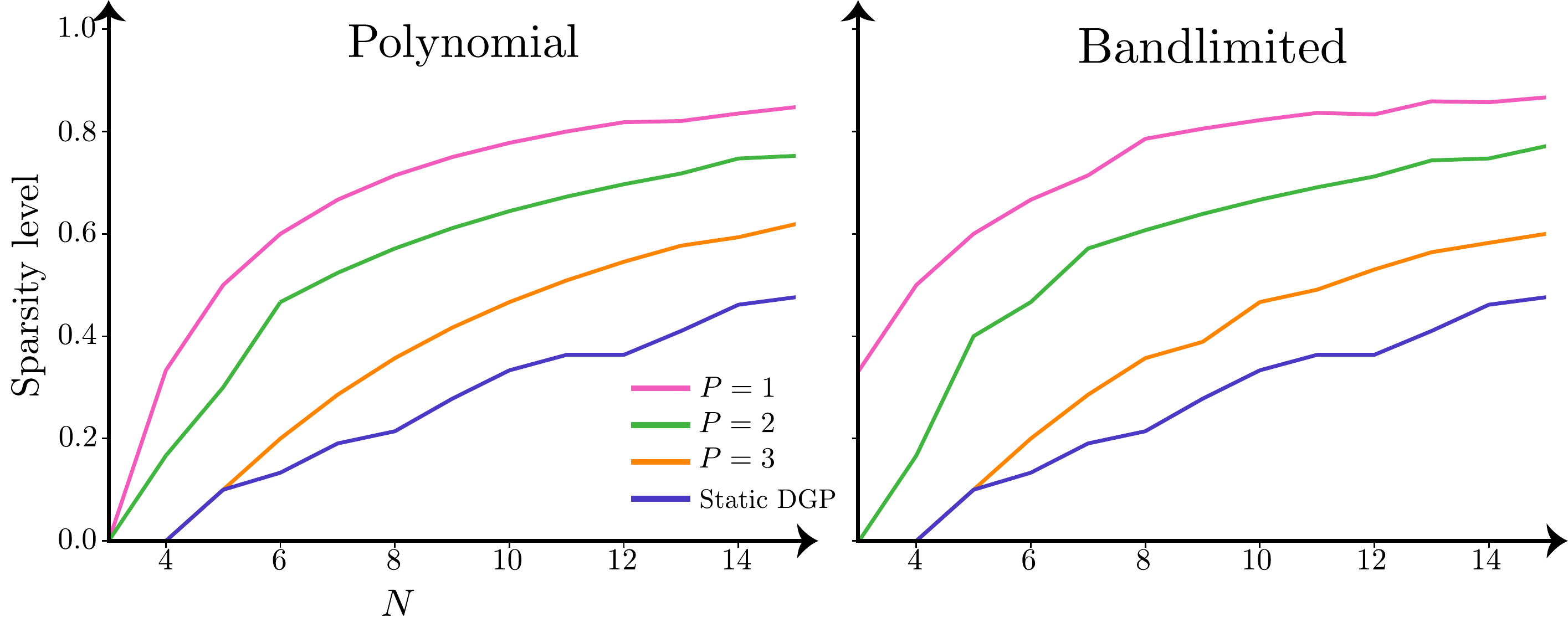}
\caption{The estimated sparsity level $\wh{S}$ for polynomial degrees $P$ and numbers of points $N$. The success threshold $\delta$ is set to $0.99$ and the target fraction of successful reconstructions $q$ to $0.9$.}\vspace{-3mm}
\label{fig:sparsity}
\end{figure}

In the first experiment, we fix the number of sampling times, $T$, and vary the number of points $N$ and polynomial (or bandlimited) degree $P$. Specifically, in \Cref{fig:sparsity} we choose $T = 7$ for polynomial and $T = 13$ for bandlimited models.

As expected, we observe that for a fixed $N$, as $P$ grows (and consequently the number of parameters) the allowable sparsity level decreases, meaning that more complicated trajectories require more spatial samples. This is due to fact that ratio of number of measurements, which is fixed in this case, to number of parameters decreases. Importantly, compared to the static DGP, we see that KEDMs and the proposed semidefinite relaxation allow us to measure fewer distances at any given time, and compensate for this by sampling at multiple times.

In the second experiment we attempt to better characterize the observed spatio--temporal sampling tradeoff. To this end, we fix the parameters so that the ratio of the number of measurements to the number of the degrees of freedom is constant. That is, we keep the number of sampling times proportional to the number of basis Gramians, $T= K+1$ for the polynomial and $T = 2K+1$ for the bandlimited model.

As \Cref{tab:sparsity_normal_poly,tab:sparsity_normal_band} show, with this scaling the sparsity level is approximately constant as the polynomial degree $P$ grows. In other words, even though the trajectories become more and more complicated, we can keep the number of spatial measurements fixed as long as we adjust the number of temporal sampling instants. The empirical observation that the required number of measurements scales linearly with the number of the degrees of freedom suggests that
the proposed algorithms require an order-optimal number of samples.

% However, there is a meaningful difference between sparsity levels for $P=0$, i.e. static DGP, and $P \neq 0$ models. For simplicity, let us compare the polynomial models with $	P=0$ and $P=1$. In the static model of $P=0$, we sample the distance matrix one time and estimate point positions at that time, i.e. estimate $\mA_0$. On the other hand, for $P = 1$ model, we sample KEDM at $K+1 = 3$ time instants to estimate $\mA_0, \mA_1$. This redundancy in parameterization of Gramian $\mG(t)$, which is due to convolution operator in \eqref{eq:g_poly_tra}, lets us achieve sparser measurements in non-trivial, $P \neq 0$, trajectory models. 

\begin{table*}[ht]
\caption{Maximal sparsity for the polynomial model and $d=2$.} 
\centering
\begin{tabular}{@{}llllllllllll@{}}
\toprule$P \backslash N$ & $5$ & $6$ & $7$ & $8$& $9$ & $10$ & $11$& $12$& $13$ & $14$ & $15$ \\
\midrule
% $P=0$ & \cellcolor{Gray}0.1 & 0.13 &\cellcolor{Gray} 0.19 & 0.21 &  \cellcolor{Gray}0.27 & 0.33 & \cellcolor{Gray}0.36 & 0.36 & \cellcolor{Gray}0.41 &   0.46 & \cellcolor{Gray}0.47\\
$P=1$ & 0.1 &  \cellcolor{Gray}0.2 & 0.28 & \cellcolor{Gray}0.39 & 0.44 & \cellcolor{Gray}0.46 & 0.52 & \cellcolor{Gray}0.56 &  0.57 & \cellcolor{Gray}0.60 &  0.62\\
$P=2$ & 0.1 &  \cellcolor{Gray}0.2 & 0.33 & \cellcolor{Gray}0.35 & 0.41 & \cellcolor{Gray}0.46 & 0.51 & \cellcolor{Gray}0.54 &  0.57 & \cellcolor{Gray}0.60 &  0.62\\
$P=3$ & 0.1 &  \cellcolor{Gray}0.2 & 0.28 & \cellcolor{Gray}0.35 & 0.41 & \cellcolor{Gray}0.46 & 0.51 & \cellcolor{Gray}0.54 &  0.57 & \cellcolor{Gray}0.59 &  0.62 \\
\bottomrule
\end{tabular}
\label{tab:sparsity_normal_poly}
\end{table*}

\begin{table*}[ht]
\centering
\caption{Maximal sparsity for the bandlimited model and $d=2$.} 
\begin{tabular}{@{}llllllllllll@{}}
\toprule
$P \backslash N$ & $5$ & $6$ & $7$ & $8$& $9$ & $10$ & $11$& $12$& $13$ & $14$ & $15$ \\
\midrule
% $P=0$ & \cellcolor{Gray}0.1 & 0.13 &\cellcolor{Gray} 0.19 & 0.21 &  \cellcolor{Gray}0.27 & 0.33 & \cellcolor{Gray}0.36 & 0.36 & \cellcolor{Gray}0.41 &   0.46 & \cellcolor{Gray}0.47\\
$P=1$ & 0.1 &  \cellcolor{Gray}0.26 & 0.33 & \cellcolor{Gray}0.39 & 0.44 & \cellcolor{Gray}0.48 & 0.51 & \cellcolor{Gray}0.56 &  0.57 & \cellcolor{Gray}0.60 &  0.63\\
$P=2$ & 0.1 &  \cellcolor{Gray}0.2 & 0.28 & \cellcolor{Gray}0.35 & 0.41 & \cellcolor{Gray}0.44 & 0.49 & \cellcolor{Gray}0.53 &  0.56 & \cellcolor{Gray}0.60 &  0.63\\
$P=3$ & 0.1 &  \cellcolor{Gray}0.2 & 0.28 & \cellcolor{Gray}0.35 & 0.38 & \cellcolor{Gray}0.46 & 0.49 & \cellcolor{Gray}0.53 &  0.56 & \cellcolor{Gray}0.58 &  0.60 \\
\bottomrule
\end{tabular}
\label{tab:sparsity_normal_band}
\end{table*}

\subsection{Noisy Measurements}\label{sec:noise}

We again quantify the influence of noise by the relative trajectory mismatch. We fix a trajectory, shown in \Cref{fig:sketch}, and a set of distance sampling times $\{ t_k \}_{k=0}^K$, and generate many realizations of noisy measurement sequences $\wt{\mathcal{D}}_{t_0}, \cdots, \wt{\mathcal{D}}_{t_K}$ with the same noise variance $\sigma^2$. The i.i.d. noise is added to the non-squared distances. The empirical trajectory mismatch is an average of relative trajectory mismatches over realizations, $\frac{1}{M}\sum_{m}{e_X^{(m)}}$.

In \Cref{fig:sketch}, we show many estimated trajectories $\wh{\mX}(t)$. As expected, the mismatch increases with measurement noise $\sigma^2$ and decreases with the number of measurements. In all cases, the estimated trajectories concentrate around the true ones. 

\begin{figure}[t]
\centering
\includegraphics[width=1\linewidth]{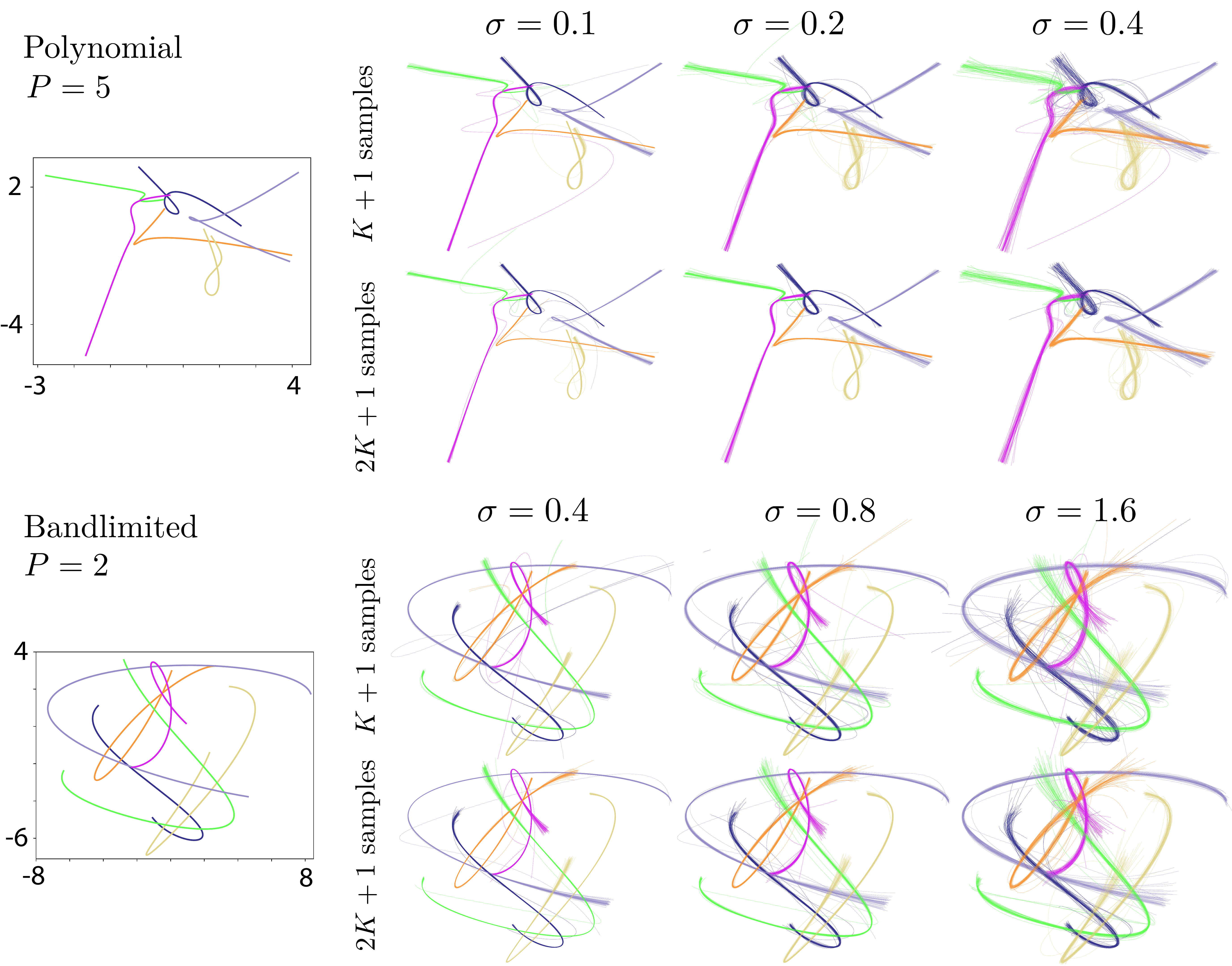} 
\caption{Estimated trajectories, $\wh{\mX}(t)$, for $N=6$ points in $\R^{2}$ at different levels of measurement noise and number of temporal measurements. The time interval of interest is $t \in [-1, 1]$ for polynomial and $t \in [0, 1]$ for bandlimited trajectories.}
\label{fig:sketch}
\end{figure}

\subsection{A Stylized Application: Satellite Positioning} \label{sec:gps}

In this section we apply KEDMs in a stylized satellite positioning scenario where measurements are both very sparse and noisy. We consider a set of satellites moving with constant angular velocity, with angular frequency being an integer multiple of the fundamental frequency $\omega_0$. This is a limitation of the current bandlimited model which we intend to address in future work. Such trajectroeis have the form
\begin{equation*}
\vx(t) = \mR \left( \begin{array}{c} a \cos(\omega t) \\ b \sin(\omega t)  \\ 0\end{array} \right),
\end{equation*}
where $\mR$ is a $3 \times 3$ rotation matrix. 

The set of all satellite trajectories,
\begin{equation*}
\mX (t) = [\vx_1(t,p_1), \cdots, \vx_{N}(t,p_N)]
\end{equation*}
follows the bandlimited trajectory model. Concretely,
\begin{equation*}
    \vx(t,p) = \va_{1} \cos(p \omega_0 t) + \va_{2} \sin(p\omega_0 t)
\end{equation*}
is the trajectory of a satellite whose angular frequency is $p$ times the fundamental frequency $\omega_0$ and $\va_1, \va_2 \in \R^{3}$. The ensemble trajectory, $\mX(t)$, is a bandlimited trajectory of degree $P = \max_{n}{p_n}$. 

We apply our algorithms \Cref{alg:kedmAL} in two experiments. In \Cref{fig:satsim1}, we show trajectories of $N=8$ satellites with the same orbiting frequency $\omega_0$. Since the ellipses are of different sizes, the inner points can also be interpreted as vehicles on the earth. We measure $3$ noisy pairwise distances, out of $28$ available, per sampling time instant. This could model, for instance, occlusions by the earth and other adversarial effects. We compensate for undersampling in space by oversampling in time, taking samples at $T = 30$ different times. Similarly, in \Cref{fig:satsim2} we show $N=5$ satellites with angular frequencies $\omega_0$ and $2\omega_0$, that is, with $P = 2$; we measure only $2$ pairwise distances per sampling time instant (these are extremely sparse measurements with which static localization is  hopeless), at $T=30$ sampling times. As figures show, in both experiments, we successfully reconstruct trajectories of the satellites. 

\begin{figure}[h]
\centering
\includegraphics[width=0.5\textwidth]{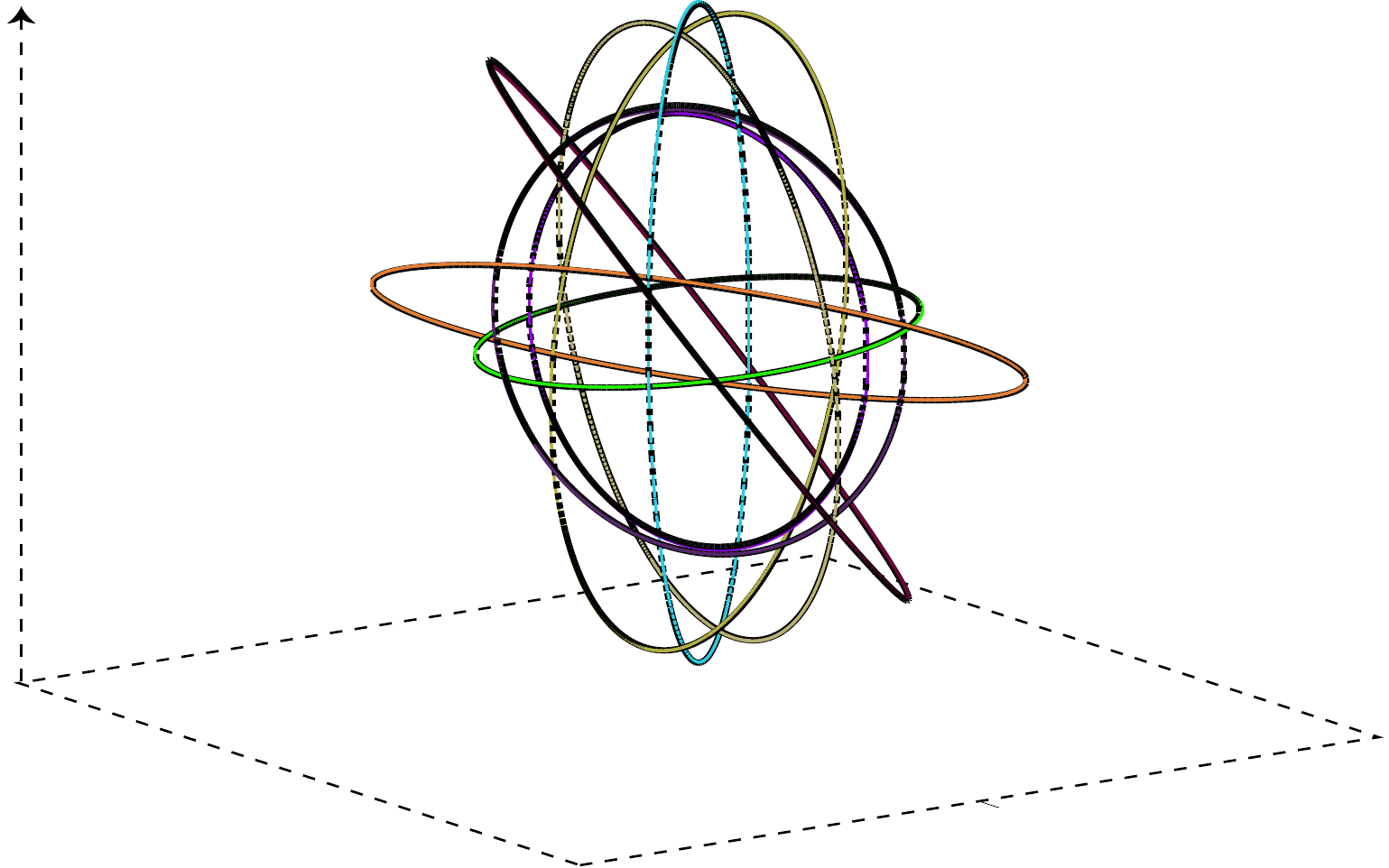}
\caption{Reconstructing the trajectories of $8$ orbiting satellites. Colored and dashed lines represent actual and estimated trajectories. All satellites have the same angular frequency with $P=1$. The measurement matrices are missing about $9/10$ measurements, and noise level is set to $\sigma = 0.05$. The average reconstruction error is $\frac{1}{M}\sum_{i}{e_D(t_i)} = 0.01$.}
\label{fig:satsim1}
\end{figure}

\begin{figure}[h]
\centering
\includegraphics[width=0.5\textwidth]{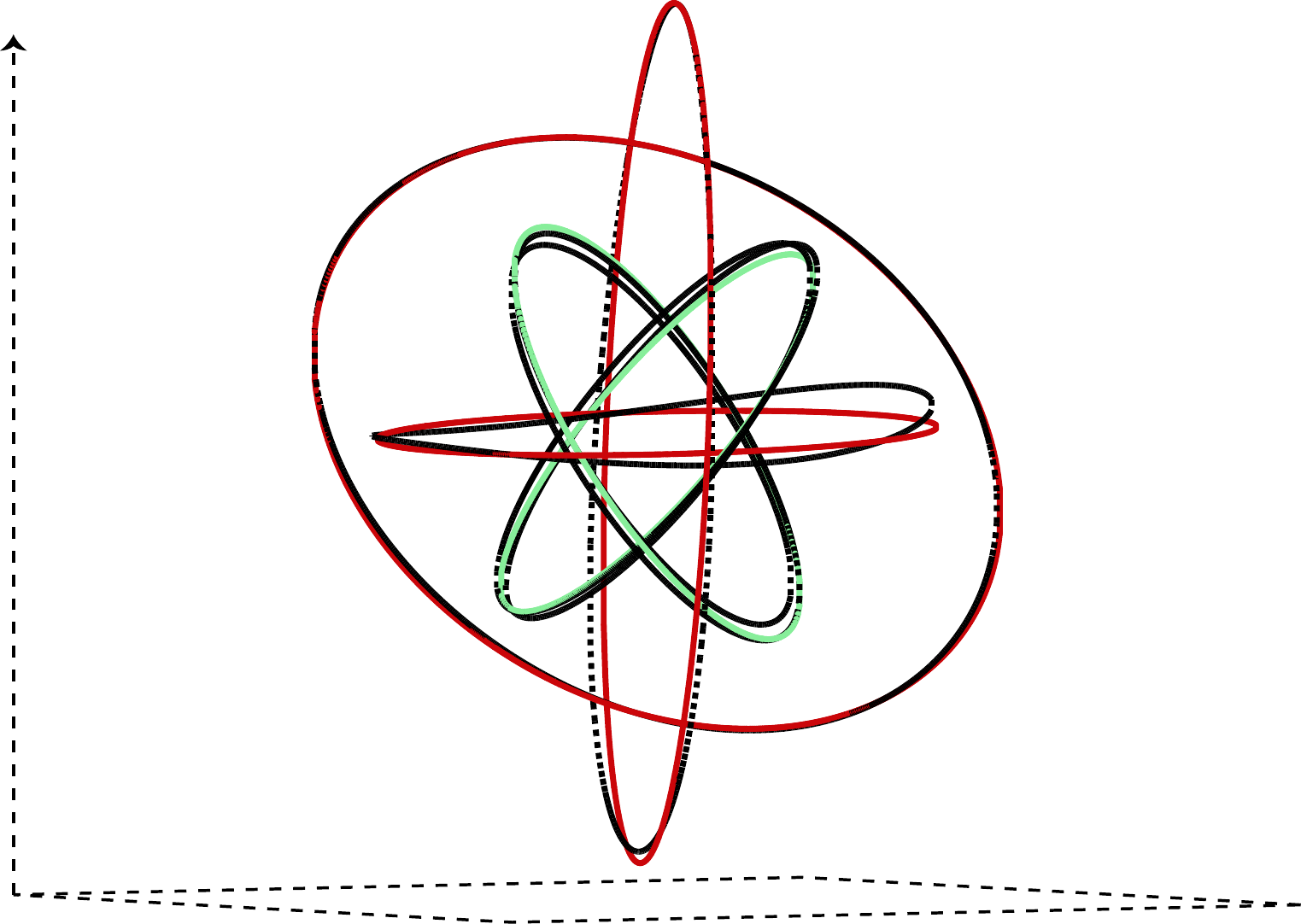}
\caption{Reconstructing the trajectories of $5$ orbiting satellites with angular frequencies of $\omega_0$ and $2\omega_0$. The measurement matrices are $80 \%$ sparse, and average reconstruction error is $\frac{1}{M}\sum_{i}{e_D(t_i)} = 0.03$.}
\label{fig:satsim2}
\end{figure}

\section{Conclusion}

In this paper, we extended the algebraic tools for localization from distances to the case when points are moving. We defined kinetic Euclidean distance matrices for polynomial and bandlimited trajectories, and we derived algorithms based on semidefinite programming to solve the associated trajectory localization problem. The chosen trajectory models are expressive and can approximate continuous trajectories commonly used in localization and tracking. 

The key step in our method is to represent the time-varying Gram matrices as time-varying linear combinations of certain constant matrices. This allowed us to rewrite the localization problem as a semidefinite program. Same as in the static case, the actual localization involves an additional spectral factorization step. However, for polynomial matrices, this is much harder than a simple SVD, and especially from noisy data like those that we get. We circumvent the related difficulties by deriving a spectral factorization method that directly uses anchor measurements.

We demonstrated through numerical experiments that the proposed algorithms can indeed reconstruct model trajectories from sparse and noisy measurements, and that they can explore the tradeoff between the number of distances measured at any given time, and the number of sampling times.

\subsection{Future Work} % (fold)

Both the polynomial and the bandlimited trajectories are special cases of a general class of subspace trajectories. Conceptually, one should be able to derive the localization theory for general subspace trajectories (for example, one could mix bandlimited and polynomial trajectories). Doing this cleanly is not trivial and is part of ongoing work.
Instead of adopting deterministing trajectory models, we could think of stochastic models. The KEDM would then become a random object and ideas of statistical inference could be used to extract point position information. Stochastic models should play particularly well with SLAM.

In the presented numerical experiments we empirically explored the tradeoff between the spatial and the temporal measurements. An interesting and important line of work is to characterize this tradeoff analytically, as a function of the algorithm used for localization.

Finally, an interesting thing happens once we depart from static EDMs: instead of only measuring inter-point distances, we can  measure relative (vector or scalar) velocities, and the same goes for anchors. It is of considerable practical interest to derive tractable optimization procedures that take into account these general kinetic measurements.

\appendices

\section{Spectral Factorization of the Gramian}

Let $q$ stand for $t$ for the polynomial model, or $e^{j\omega t}$ for the bandlimited model. Similarly, let $\mathcal{P} = \set{0, \ldots, P}$ for polynomial or $\mathcal{P} = \set{-P, \ldots, P}$ for bandlimited.

% \footnote{A Laurent polynomial with coefficients in a field $\mathbb{F}$, is expressed as $x(z) = \sum_{p}{c_p z^p}$ where $z$ is a formal variable and can have negative powers. Bandlimited trajectories are a special case of Laurent polynomials where $\mathbb{F}= \C^{d \times N}$ and $z = e^{j \omega}$.}

\begin{lemma}\label{lem:specfactor}
Let $\mG(q) = \sum_{p \in \mathcal{P} + \mathcal{P}} \mB_p q^p$ (with $\mB_p \in \C^{N \times N}$) be rank-$d$ and positive semidefinite. Then there exists a unique (up to a $d \times d$ left unitary factor) $d \times N$ matrix $\mX(q) = \sum_{p \in \mathcal{P}}{\mA_k q^p}$ such that $\mG(q) = \mX(q)^{H}\mX(q)$.
\end{lemma}

The statement has been proved for Laurent matrix polynomials in \cite{ephremidze2015rank}. For $q = t$ it is equivalent to spectral factorization of polynomial matrices on the real line. Ephremidze \cite{ephremidze2014elementary} proved the full rank version; an entirely parallel construction to those in \cite{ephremidze2014elementary,ephremidze2015rank} implies that it holds of rank-deficient matrices.

\section*{Acknowledgment}

The authors would like to thank Lasha Ephremidze for help with understanding spectral factorization.

% Can use something like this to put references on a page
% by themselves when using endfloat and the captionsoff option.
\ifCLASSOPTIONcaptionsoff
  \newpage
\fi

\bibliographystyle{IEEEtran}
\bibliography{kedm}

% biography section
% 
% If you have an EPS/PDF photo (graphicx package needed) extra braces are
% needed around the contents of the optional argument to biography to prevent
% the LaTeX parser from getting confused when it sees the complicated
% \includegraphics command within an optional argument. (You could create
% your own custom macro containing the \includegraphics command to make things
% simpler here.)
%\begin{IEEEbiography}[{\includegraphics[width=1in,height=1.25in,clip,keepaspectratio]{mshell}}]{Michael Shell}
% or if you just want to reserve a space for a photo:

% \begin{IEEEbiography}{Michael Shell}
% Biography text here.
% \end{IEEEbiography}

%\vfill

% Can be used to pull up biographies so that the bottom of the last one
% is flush with the other column.
%\enlargethispage{-5in}

% that's all folks
\end{document}